\documentclass[final]{siamltex}

\usepackage[pdftex]{graphicx,color}
\usepackage{xspace}
\usepackage{amsmath,amssymb}
\usepackage{array}
\usepackage{fullpage}
\usepackage{picinpar}

\newcommand{\Rp}{\mathbb{R}_{\ge 0}}

\newcommand{\leftside}{\text{\sf left}}
\newcommand{\rightside}{\text{\sf right}}

\newcommand{\refLP}{(\ref{eq:LP})\xspace}

\newcommand{\transposal}[2]{{#1}^{#2}}
\newcommand{\triangulated}[2]{{\widehat{#1}}^{#2}}

\newcommand{\LK}{\text{\sc lk}}
\newcommand{\MCP}{\text{\sc mcp}}
\newcommand{\MWT}{\text{\sc mwt}}

\newcommand{\mwt}{\MWT}

\newcommand{\dg}{^{\circ}}

\renewcommand{\overline}[1]{{#1}}

\newcommand{\blanket}{B} 
\newcommand{\blanketSet}{{\cal B}} %
\newcommand{\convPart}{\text{\sc cp}} 
\newcommand{\cost}{c}
\newcommand{\edge}{e}
\newcommand{\edges}{E}
\newcommand{\face}{f}
\newcommand{\graph}{G}
\newcommand{\polygon}{P}
\newcommand{\region}{R}
\newcommand{\tri}{t}  
\newcommand{\vertex}{v}
\newcommand{\vertices}{V}
\newcommand{\fracTriang}{X}

\newcommand{\sensitivity}{\sigma}

\newcommand{\YXY}{YXY\xspace} 

\newcommand{\xfigpdf}[1]{\input{pictures/#1.pdf_t}}
\newcommand{\xfig}[1]{\xfigpdf{#1}}

\newtheorem{fact}{Fact}[theorem] 
\newcommand{\qed}{\endproof}

\newenvironment{proofidea}{\par{\it Proof idea}. \ignorespaces}{}

\newenvironment{fullproof}{\par{\it Full proof}. \ignorespaces}{\endproof}

\title{On a Linear Program for Minimum-Weight Triangulation
\thanks{To appear in SICOMP. Extended abstract paper appeared in SODA 2012 \cite{yousefi2012linear,yousefi2013linear}.}}

\author{Arman Yousefi\thanks{%
    University of California, Los Angeles, USA.
    Research partially funded by GAANN fellowship.
}
\and
Neal E. Young\thanks{%
  University of California, Riverside, USA.
  Research partially funded by NSF grants 0729071 and 1117954.
}
}

\begin{document}


\maketitle

\begin{abstract}
Minimum-weight triangulation (MWT) is NP-hard.
It has a polynomial-time constant-factor approximation algorithm,
and a variety of effective polynomial-time heuristics that, for many instances, can find the exact MWT.
Linear programs (LPs) for MWT are well-studied,
but previously no connection was known between any LP
and any approximation algorithm or heuristic for MWT.
Here we show the first such connections:
for an LP formulation due to Dantzig et al.~(1985):
(i) the integrality gap is constant;
(ii) given any instance, if the aforementioned heuristics find the MWT, then so does the LP.
\end{abstract}


\section{Introduction}

In 1979, Garey and Johnson listed minimum-weight triangulation (MWT) 
as one of a dozen important problems neither known to be in P 
nor known to be NP-hard \cite{garey1979computers}.
In 2006 the problem was finally shown to be NP-hard \cite{mulzer2008minimum}.
The problem has a sub-exponential time exact algorithm \cite{smith1988studies},
as well as a polynomial-time approximation scheme (PTAS) 
for random inputs \cite{golin1996limit}.
It is not currently known whether, for some constant $c>1$,
finding a $c$-approximation is NP-hard,
but this is unlikely as a quasi-polynomial-time approximation scheme exists \cite{remy2009quasi}.
MWT has an $O(\log n)$-approximation algorithm \cite{plaisted1987heuristic},
and, most important here,
an $O(1)$-approximation algorithm
called {\sc QuasiGreedy} \cite{krznaric1998quasi}.
The constant in the big-O upper bound from \cite{krznaric1998quasi}
is astronomically large.

If restricted to simple polygons, MWT
has a well-known $O(n^3)$-time dynamic-programming algorithm 
\cite{gilbert1979new,klincsek1980minimal}.
Polynomial-time algorithms also exist for instances with a constant number of ``shells''
\cite{anagnostou1993polynomial}
and for instances with only a constant number of vertices in the interior of the 
region $\region$ to be triangulated
\cite[\S 2.5.1]{Giannopoulos:2007cx},
\cite{hoffmann2006minimum,borgelt2008fixed,spillner5faster,knauer2006fixed}.

\paragraph{Linear program of Dantzig et al.~for MWT}
Linear programming (LP) methods are a primary paradigm
for the design of approximation algorithms.
For many hard combinatorial optimization problems, 
especially so-called packing and covering problems,
the polynomial-time approximation algorithm with the best approximation ratio
is based on linear programming,
either via randomized rounding or the primal-dual method.
The design of a good approximation algorithm 
is often synonymous with bounding the integrality gap
of an underlying LP.

MWT has several straightforward linear programming (LP) relaxations.
Studying their integrality gaps
may lead to better approximation algorithms,
or may widen our understanding 
of general methods and their limitations
(as standard randomized rounding and primal-dual approaches 
may be insufficient for MWT).

Dantzig et al.~(1985) introduce the following LP
(presented here as reformulated by \cite{de1996polytope}).
Below $\triangle$ denotes the set of empty triangles.\footnote
{That is, triangles lying in the region to be triangulated, 
whose vertices are in the given set of points,
but otherwise contain none of the given points.}
$\region$ denotes the region to be triangulated,
minus the sides of triangles in $\triangle$.
The LP asks to assign a non-negative weight $\fracTriang_\tri$ to
each triangle $\tri\in\triangle$ so that, for each point $p$ in the 
region, the triangles containing it are assigned total weight 1:
\begin{equation}\label{eq:LP}
  \text{ minimize }~
  \cost(\fracTriang) = \sum_{\tri\in\triangle} \cost(\tri) \fracTriang_\tri \text{{}, ~s.t. }
  \fracTriang\in\Rp^\triangle \text{ ~and~ }
  (\forall p\in \region)~\sum_{\tri\ni p} \fracTriang_\tri = 1.
\end{equation}
Above, $\Rp^\triangle$ is the set of vectors of non-negative reals, 
indexed by the triangles in $\triangle$.
The cost $\cost(\tri)$ of triangle $\tri$
is the sum over the edges $e$ in $\tri$ of the cost $\cost(\edge)$ of the edge,
defined to be $|\edge|/2$ (the length of $\edge$), unless $\edge$ is on the boundary of $\region$,
in which case the cost is $|\edge|$.
(Internal edges are discounted by 1/2 since any internal edge occurs in either zero or two triangles in any triangulation.)
$\region$ as specified is infinite, 
but can easily be restricted to a polynomial-size set of points
without weakening the LP.  (E.g., let $\region$ contain, 
for each possible edge $\edge$, two points $p$ and $q$, each on one side of $\edge$ and very near $\edge$.)

For the simple-polygon case, the above LP finds the exact MWT (every extreme point
has 0/1 coordinates, and so corresponds to a triangulation).
This was shown by Dantzig et al.~(1985)  \cite[Thm.~7]{dantzig1985triangulations},
then (apparently independently) by De Loera et al.~(1996) \cite[Thm.~4.1(i)]{de1996polytope}
and Kirsanov (2004) \cite[Cor.~3.6.2]{kirsanov2004minimal}.
For summaries of these results, see \cite[Ch.~8]{de2010triangulations} and \cite{takeuchi1998polytopes}.
Kirsanov describes an instance (a 13-gon with a point at the center) 
for which this LP has integrality gap just above 1,
as well as instances (50 random points equidistant from a center point)
that are solved by the LP but not by the LMT-skeleton heuristic.

Other authors have considered {\em edge-based} LP's, mainly for use in branch-and-bound
\cite{kyoda1996study,kyoda1997branch,ono1996package,tajima1998optimality,aurenhammer2000optimal}.
These edge-based LPs have unbounded integrality gaps.

LPs for maximal independent sets, which are well studied, are closely related to all the above LPs,
as triangulations can be defined as maximal independent sets of triangles (or of edges).
The above LPs enforce some, but not all, well-studied inequalities for maximal independent sets.

It is known to be NP-hard to determine whether there exists a triangulation
whose edge set is a subset of a given set $S$
\cite{lloyd77triangulations}.
For a given set $S$, if we change the cost function in the above LP to
$\cost(\fracTriang) = \sum_{\edge\notin S} \sum_{\tri\ni \edge} \fracTriang_\tri$,
the LP will have a zero-cost integer solution iff there is such a triangulation.
Unless P=NP, this implies that the LP with that cost function has unbounded integrality gap.\footnote
{If the LP has bounded integrality gap,  it has a zero-cost fractional solution iff it has a zero-cost integer solution.}
Thus, any bound on the integrality gap of the LP 
with the MWT cost function must rely intrinsically on that cost function.
Similarly,  given an arbitrary fractional solution $X$, it is NP-hard to determine whether there is an integer solution in the support of $X$.\footnote
{Given an arbitrary subset $S$ of the edges, the problem of determining whether $S$
contains a triangulation reduces to the problem of determining whether there is an integer solution in the support of a given fractional solution to the LP, as follows.  
Let set $S'$ consist of the empty triangles whose edges are in $S$,
so $S$ contains a triangulation (by edges) 
iff $S'$ contains a triangulation (by triangles).
For each triangle $\tri\in S'$, solve the LP with the cost function
that gives $\tri$ cost zero, every other triangle in $S'$ cost one,
and all triangles not in $S'$ cost infinity.  
If (for any $\tri\in S'$) the LP for $\tri$ has no finite-cost feasible solution,
then $S'$ contains no triangulation.  Otherwise, for each $\tri\in S'$,
let $X^\tri$ denote an optimal fractional solution to the LP for $\tri$.  
Let $\widetilde X = \sum_{\tri\in S'} X^\tri/|S'|$ be the average of these fractional solutions.
Because of the choice of the cost function, if a given $X^\tri$ does not give positive weight to $\tri$, then no (integer) triangulation in $S'$ contains $\tri$.  Thus, $S'$ contains a triangulation iff there is a triangulation in the support of $\widetilde X$.}
These are obstacles to standard randomized-rounding methods.


\paragraph{First new result: integrality gap is constant}
We show that LP \refLP\ has constant integrality gap.  
This is the first non-trivial upper bound on the integrality gap of any MWT LP.
To show it, we revisit the analysis of {\sc QuasiGreedy} \cite{krznaric1998quasi},
which shows that {\sc QuasiGreedy} produces a triangulation
of cost $O(|\mwt(\graph)|)$, where $|\mwt(\graph)|$ is the length of the MWT 
of the given instance $\graph$
(and also the cost of the optimal integer solution to the LP).
We generalize their arguments to show that there exists a triangulation
of cost $O(\cost(\fracTriang^*))$, where $\cost(\fracTriang^*)$ is the cost of the optimal
{\em fractional} solution to the LP.

Our analysis also reduces the approximation ratio in their analysis by an order of magnitude, 
but the approximation ratio remains a large constant.

\paragraph{MWT heuristics}
Much of the MWT literature concerns polynomial-time heuristics that, given an instance, 
find edges that must be in (or excluded from) any MWT.
Here is a summary.
Gilbert observes that the {\em shortest potential edge} is in every MWT \cite{gilbert1979new}.
Yang et al.~extend this result by proving that an edge $\overline{xy}$ is in every MWT 
if, for any edge $\overline{pq}$ that intersects $\overline{xy}$,
$|\overline{xy}| \le \min\{|\overline{px}|,|\overline{py}|,|\overline{qx}|,|\overline{qy}|\}$ \cite{yang1994chain}.
(We refer to the edges satisfying this property as the {\em $\YXY$ subgraph}.)
This subgraph includes every edge connecting two {\em mutual nearest neighbors}. 
Keil \cite{keil1994computing} defines another heuristic called {\em $\beta$-skeleton} as follows. 
An edge $\overline{p q}$ is in the $\beta$-skeleton if and only if there does not exist a point $z$
in the point set such that $\angle p z q \ge \arcsin(1/\beta)$. 
Thus, an edge $\overline{p q}$ is in the $\beta$-skeleton if and only if 
the interior of the two circles of diameter 
$\beta\cdot|{\overline{p q}}|$ passing through $p$ and $q$ do not contain any points.
Keil \cite{keil1994computing} then shows that for $\beta \ge \sqrt{2}$, an edge that is in 
$\beta$-skeleton is in every MWT. 
Cheng et al.~strengthen this to $\beta \ge 1/\sin k $ where $k \approx \pi/3.1$ \cite{cheng1996approaching}.
Das and Joseph show that an edge $\edge$ cannot be in any MWT if
both of the two triangles with base $\edge$ and base angle $\pi/8$ 
contain other vertices \cite{das1989triangulations}. 
Drysdale et al.~strengthen this to angle $\pi/4.6$ \cite{drysdale2001exclusion}.
This property of $\edge$ is called the {\em diamond property}.
Dickerson et al.~describe a simple local-minimality property such that,
if an edge $\edge$ lacks the property,
the edge cannot be in any MWT.
Using this, they show that the so-called
{\em LMT skeleton} must be {\em in} the MWT
\cite{dickerson1997large}.

A primary use of the heuristics is to solve
some instances of MWT exactly in polynomial time, as follows:
{\em Given an instance, use the heuristics to identify edges that are in the MWT.
If the regions left untriangulated by these edges are simple polygons
(equivalently, if the edges span the given points)
then find the MWT of each region independently
using the standard dynamic programming algorithm.}
(The MWT will be the union of the MWT's of the regions.)
According to \cite{dickerson1997large} (1997),
most random instances with 40,000 points are solvable in this way.

\paragraph{Second new result: LP generalizes heuristics}
We show that 
LP~\refLP\ generalizes these heuristics
in that {\em if the heuristics solve a given instance
as described above, then so does the LP}
(that is, the extreme points of the LP are integer solutions --- incidence vectors of optimal triangulations).
In this sense, the LP, whose formulation requires little explicit geometry,
generalizes all of these varied and generally incomparable heuristics.
(In fact the LP appears to be stronger than the heuristics, in that some natural instances
are solved by the LP, but not by the heuristics
\cite[\S 3.5]{kirsanov2004minimal}.\footnote
{Where $\graph$ contains the center of a unit circle and $n-1$ random points on the circle.})
This is the first connection we know of between the heuristics and any MWT LP.

Roughly, the heuristics are based on a combination of
(i) local-improvement arguments about the MWT and
(ii) logical closure (once the heuristic determines the status of one edge with respect to the MWT,
this in turn determines the status of other edges, and so on).
We extend these arguments to apply to the optimal fractional triangulation $\fracTriang^*$.
This is possible because
(i) $\fracTriang^*$  looks ``locally'' like a MWT and 
(ii) the LP enforces logical closure of linear constraints on $\fracTriang^*$.

After we finished the body of this work, we became aware of
and examined additional heuristics by Wang et al.~\cite{wang1997new} 
and Aichholzer et al.~\cite{aichholzer1996triangulations}.
We conjecture that the LP generalizes them as well.

\paragraph{An equivalent formulation of the LP}
The following constraints are equivalent to the last constraints in LP~\refLP
(see e.g.~\cite[Thm.~1.1(i), Prop.~2.5]{de1996polytope},
\cite{takeuchi1998polytopes}, or \cite[Thm.~3.4.1]{kirsanov2004minimal})
and are useful for reasoning about fractional triangulations.
For any fractional triangulation $\fracTriang$ and edge $\edge$,
\begin{equation} \label{constraint:edge}
\sum_{\tri\in \leftside(\edge)} \fracTriang_\tri ~-~ \sum_{\tri\in \rightside(\edge)} \fracTriang_\tri ~=~ [\edge\in \text{boundary}(\region)].
\end{equation}
Here $\leftside(\edge)$ contains the triangles that contain $\edge$ and lie on one side of $\edge$,
while $\rightside(\edge)$ contains the triangles that contain $\edge$ and lie on the other side of $\edge$.
(If $\edge$ is on the boundary, take $\rightside(\edge)=\emptyset$.)
The notation $[x\in S]$ denotes 1 if $x\in S$ and 0 otherwise.

\paragraph{Practical considerations}
Using the $O(n^2)$ constraints~(\ref{constraint:edge})  
instead of the constraints in~\refLP gives an equivalent LP
with total size (i.e., non-zeros in the constraint matrix) 
proportional to the number of empty triangles.
The empty triangles can be identified, and the LP constructed, 
in time proportional to their number \cite{dobkin1990searching}.
Their number is always $O(n^3)$, 
but often smaller (e.g.~$O(n^2)$ in expectation for randomly distributed points).

The time to construct and solve the LP can be further reduced 
by a preprocessing step based on the heuristics
--- 
remove any variable $\fracTriang_\tri$ if the heuristics prove any edge of $\tri$ to be excluded from every MWT,
and add a constraint
$\sum_{\tri\in \leftside(\edge)} \fracTriang_\tri 
= \sum_{\tri\in \rightside(\edge)} \fracTriang_\tri = 1$ 
if they prove an interior edge $\edge$ to be in every MWT.
For randomly distributed points,
only $O(n)$ edges (in expectation) have the diamond property,
forming $O(n^2)$ possible empty triangles,
from which the modified LMT skeleton can be computed in $O(n^2)$ time
\cite{dickerson1997large,dickerson1997fast}.
In our ad-hoc experiments on ``typical'' instances with $10^4-10^5$ points,
only a small number of variables were left undetermined by the heuristics.
This allowed us to use standard LP solvers to quickly solve the remaining LP.
(This is in keeping with Dickerson et al's experiments, which found that most random
instances on 40,000 points were solvable by heuristics \cite{dickerson1997large}.)
Similarly, this preprocessing
should help integer-LP solvers to quickly find the MWT
(for instances for which the optimal solution is fractional).
It is known that, asymptotically, 
for $n$ random points, the expected number of remaining variables
is $\Omega(n)$, but the leading constant is apparently astronomically small
\cite{bose2002diamonds}.


\paragraph{Remarks} 
The results here suggest that the LP of Dantzig et al.\ captures
much of the structure of MWT.
This suggests a line of attack for improving the approximation ratio:
use systematic LP methods such as randomized rounding, the primal-dual method,
and lift-and-project \cite{balas2002lift} to study the integrality gap of the LP.
Success would yield an better approximation (conceivably, even a PTAS, using lift-and-project).
Failure would increase our understanding of the limitations of these techniques.

Implicit in our bound on the integrality gap
is a polynomial-time algorithm with matching approximation ratio.
Actually, there are two.
Both algorithms first compute Levcopoulos and Krznaric's convex partition $\LK$
of the point set
(see our Lemma~\ref{lemma:LK}) \cite{krznaric1998quasi},
then extend $\LK$ by triangulating each face $\face$ of $\LK$.
The triangulation of each $\face$ can be done either 
(a) using the standard dynamic program to find a minimum-weight triangulation of $\face$,
or (b) as follows: compute the fractional solution $\fracTriang$ to the linear program,
then, for each face $\face$ of $\LK$,
{\em transpose} $\fracTriang$ into a fractional triangulation $\transposal{\fracTriang}{\face}$
of $\face$
(as described in Section~\ref{sec:part i}),
then use the cheapest triangulation of $\face$ implicit in $\transposal{\fracTriang}{\face}$.

That the first algorithm above is an $O(1)$-approximation algorithm
follows from Levcopoulos and Krznaric's previous work \cite{krznaric1998quasi}.
However, the bound we show here
--- $54(\lambda+1)$, where $\lambda$ is a large constant per Lemma~\ref{lemma:LK} ---
is substantially smaller 
than their previous bound.
Roughly, we obtain a better bound by analyzing the transposal operation 
at the level of {\em triangles}, instead of edges.

\paragraph{Open problems}
The integrality gap is constant,
but there is still a huge gap between the best lower bound known (barely above 1.0) and the upper bound shown here (astronomically large).
The next step in improving our upper bound
would be to reduce the value of $\lambda$ in Lemma~\ref{lemma:LK}.
We suspect that a primal-dual analysis is implicit in the analysis here; 
making the dual solution explicit might be a step in this direction.

Many different cost functions 
(other than the total edge length) for triangulations
are studied in the literature.
The MWT LP extends naturally
by modifying the cost function or restricting the set of allowed triangles.
(For example, the integrality of the extreme points of the LP for the simple-polygon case
implies that the simple-polygon result generalizes to any linear cost function.)
We conjecture that results similar to those in this paper can be obtained for other cost functions.

If MWT heuristics can solve a given instance of MWT, then so can the LP.
However, the heuristics are also useful for instances that they don't completely solve:
on such instances, the heuristics can still identify some edges that are in (or excluded from) 
every MWT, even if these do not completely determine the triangulation.
Can some analogous property be shown for the LP?
That is, is there some condition (e.g., based on the optimal primal/dual solution to the LP) such that,
if the condition holds for an edge $\edge$, then that edge must be in (or excluded from) every MWT?


\smallskip
\begin{definition}
The {\em interior} of a segment $\overline{p q}$ is $\overline{p q}-\{p,q\}$.
The {\em interior} of a polygon $\polygon$ consists of $\polygon$ minus its boundary.
Two sets {\em properly intersect} (or {\em overlap}, or {\em cross}) 
if the intersection of their interiors is non-empty.
The (Euclidean) length of line segment $\overline{p q}$ is denoted $|\overline{p q}|$.
For any set $\edges$ of segments, 
$\|\edges\|_2$ denotes the total length of segments in $\edges$. 

A {\em planar straight-line graph} (PSLG) is an undirected graph $\graph=(\vertices,\edges)$
along with a planar embedding that identifies each vertex with a planar point
and each edge with the line segment connecting its endpoints,
so that each edge intersects other edges (and $\vertices$) only at its endpoints.
The {\em length} of $\graph$ is the sum of the Euclidean lengths of its edges.
$\graph$ partitions the plane into polygonal {\em faces}.\footnote
{Where two points are in the same face if there is a path between them that intersects no edge,
with the caveat that the term {\em face} excludes the single such unbounded region.}
A face or polygon is {\em empty} if its interior contains no vertex.

A {\em diagonal}, or~{\em potential edge}, of $\graph$ 
is any segment $\overline{p q}\not\in \edges$ connecting two vertices of a face,
and contained in that face, so that $\graph'=(\vertices,\edges\cup\{\overline{p q}\})$ is still a PSLG.
A {\em partition} of $\graph$ is a PSLG that extends $\graph$ by adding (non-crossing) diagonals;
equivalently, the faces of the partition refine the faces of $\graph$.
A {\em convex partition} of $\graph$ is a partition whose faces are empty and strictly convex.
The minimum-length convex partition of $\graph$ is denoted $\MCP(\graph)$.
A {\em triangulation} of $\graph$ is a partition whose faces are empty triangles.
A {\em fractional triangulation} $\fracTriang$ is a feasible solution to the LP.
For any potential edge $\edge$, the {\em weight of edge} $\edge$ in $\fracTriang$, denoted $\fracTriang_\edge$, 
is $\sum_{\tri\ni\edge} \fracTriang_\tri$ if $\edge$ is on
the boundary of the region to be triangulated, and otherwise half this amount.


Formally, an instance of MWT is specified by a planar point set $\vertices$,
implicitly defining a PSLG $\graph=(\vertices,\edges)$ where $\edges$ contains the edges 
on the boundary of the convex hull of $\vertices$.
A solution is a minimum-length triangulation of $\graph$.
\end{definition}

Throughout, we fix an instance $\graph=(\vertices,\edges)$ of MWT specified by a given point set $\vertices$.
Unless stated otherwise, every graph considered is a partition of $\graph$.
Since the vertex set $\vertices$ is the same for all such graphs, we identify each particular graph by its edge set.


\section{Integrality gap is constant}  

This section proves our first new result:

\begin{theorem}\label{thm:gap}
  Given any instance $\graph=(\vertices,\edges)$ of MWT, 
  for any fractional triangulation $\fracTriang$,
  there exists an integer solution of value $O(\cost(\fracTriang))$.
  That is, LP~\refLP has constant integrality gap.
\end{theorem}
\proof
Fix the MWT instance $\graph$ and an arbitrary fractional triangulation $\fracTriang$.
Fix a convex partition $\convPart$ of $\graph$.
(Later, we will fix $\convPart$ to be a particular convex partition $\LK$ with some particular properties.)

The idea of the proof is to define a ``rounding'' procedure that 
converts $\fracTriang$ into the desired integer solution.
The procedure fractures $\fracTriang$ into 
a separate fractional triangulation $\transposal{\fracTriang}{\face}$ 
for each face $\face$ of $\convPart$
(where $\transposal{\fracTriang}{\face}$ covers exactly $\face$).
Then, independently within each face $\face$ of $\convPart$,
the procedure replaces the fractional triangulation $\fracTriang^\face$ 
by the optimal integer triangulation of $\face$.
The final ``rounded'' solution is then the union
of these integer triangulations (one for each face $\face$ of $\convPart$),
of total cost at most $\sum_{\face\in \convPart} \cost(\transposal{\fracTriang}{\face})$ (and, hopefully, $O(\cost(\fracTriang))$).

In the second step, since each $\face$ is a simple polygon, it follows from 
known results (e.g.~\cite[Thm.~7]{dantzig1985triangulations}; see the introduction)
that the cost of the optimal integer triangulation of $\face$ is at most the cost of $\fracTriang^\face$.
Thus, the integrality gap will be $O(1)$
as long as the first step triangulates the faces so that
$\sum_\face \cost(\transposal{\fracTriang}{\face}) = O(\cost(\fracTriang))$.

The proof divides into two parts:
(i) describing a correct rounding procedure that fractures $\fracTriang$ 
into a fractional triangulation $\transposal{\fracTriang}{\face}$ 
for each face $\face$ of $\convPart$
(we call this {\em transposing} $\fracTriang$ into $\face$)
and
(ii) bounding the cost $\sum_\face \cost(\transposal{\fracTriang}{\face})$ by $O(\cost(\fracTriang))$.

\subsection{Part (i) --- fracturing $\fracTriang$ into the faces of $\convPart$}
\label{sec:part i} \label{sec:feasibility}
Fix any face $\face$ of $\convPart$ of the convex partition $\convPart$.
Our goal is to convert $\fracTriang$ into a fractional triangulation
$\transposal{\fracTriang}{\face}$ of $\face$.

We start with the observation that $\fracTriang$, restricted to triangles that cross $\face$,
can be separated into independent layers, 
where each layer is a set of triangles that uniformly covers $\face$
(and possibly some points outside $\face$).
We say such a layer {\em blankets} $\face$:
\begin{definition}[blanket]\label{def:blanket}
  A set $\blanket$ of empty polygons with endpoints in $\vertices$
  {\em blankets} the face $\face$ 
  if the union of the polygons contains $\face$
  and no two of the polygons overlap within $\face$
  (they may overlap outside $\face$).
  {\em (In this subsection, the polygons in blankets are always triangles.)}
\end{definition}

The next lemma describes how to decompose $\fracTriang$ (over $\face$) into blankets:
\begin{lemma}\label{lemma:blankets}
  There exists a set $\blanketSet$ of blankets (each containing only triangles)
  and weights $\epsilon_\blanket>0$ for each $\blanket\in\blanketSet$,
  such that $\sum_{\blanket\in\blanketSet} \epsilon_\blanket = 1$ and, 
  for every triangle $\tri$ crossing $\face$,
  \( \fracTriang_\tri ~=~\sum_{\blanket\in\blanketSet} \,[\tri\in\blanket]\,\epsilon_{\blanket}.\)
  {\em (Recall ``$[\tri\in\blanket]$'' is 1 if $\tri$ is in $\blanket$, else 0.)}
\end{lemma}
\begin{proof}
  Recall that, for MWT instances consisting of a simple polygon, 
  the LP gives optimal 0/1 solutions (e.g.,~\cite[Thm.~7]{dantzig1985triangulations}).
  We adapt a proof of that property.

  Choose any triangle $\tri$ that crosses $\face$ and has $\fracTriang_\tri>0$.
  If $\tri$ completely covers $\face$, then stop and take $\blanket = \{\tri\}$.
  Otherwise, some edge $\edge$ of triangle $\tri$ crosses the interior of $\face$.
  Since $\edge$ has positive weight, there must be a positive-weight triangle $\tri'$ 
  that has $\edge$ as an edge and lies on $\edge$'s opposite side
  (this is implied by Constraint (\ref{constraint:edge})).
  Glue $\tri$ and $\tri'$ together to form a polygonal region.
  Continue in this way, growing the polygonal region by
  repeatedly gluing a new triangle to any boundary edge $\edge$ that crosses $\face$.
  Stop when the region has no boundary edge that crosses $f$.
  The triangles glued together in this way form the blanket $\blanket$.

  Let $\epsilon_\blanket$ be the minimum weight of any triangle in $\blanket$.
  This gives the first blanket $\blanket$ and its weight $\epsilon_\blanket$.
  Subtract $\epsilon_\blanket$ from each $\fracTriang_\tri$ for $\tri\in\blanket$.
  This reduces $\fracTriang$'s coverage of $\face$ uniformly by $\epsilon_\blanket$.
  To generate the remaining blankets in $\blanketSet$ (and their weights),
  iterate this process as long as $\fracTriang$ still covers $\face$ 
  with positive (and necessarily uniform) weight.
  (The process does terminate, as each iteration brings some $\fracTriang_\tri$ to zero.)
\end{proof}

Fix the set $\blanketSet$ of blankets of $\face$ from Lemma~\ref{lemma:blankets}
and the corresponding weights $\epsilon_\blanket$.

We next describe how to convert any single blanket $\blanket\in\blanketSet$
into a true triangulation $\triangulated{\blanket}{\face}$ of $\face$.  
The final fractional triangulation
$\transposal{\fracTriang}{\face}$ will be the convex combination of these
triangulations, where the triangulation of $\blanket$ is given weight $\epsilon_\blanket$.

Recall that any blanket $\blanket\in\blanketSet$ consists of triangles
that together uniformly cover the convex face $\face$ (and may extend outside of $\face$).
To define the triangulation $\triangulated{\blanket}{\face}$,
we start with {\em edge transposals} \cite[e.g.~Lemma 4.2]{krznaric1998quasi}.
For any edge $\edge$ that crosses $\face$, 
{\em transposing} $\edge$ in $\face$
slides $\edge$ to its {\em transposal}, denoted $\transposal{\edge}{\face}$,
a diagonal of $\face$ that has minimum length 
among four or fewer diagonals that are ``near'' $\edge$.
We give the formal definition next,
and then extend that to define transposals 
of triangles and blankets $\blanket\in\blanketSet$.
\begin{definition}[transposing an edge \cite{krznaric1998quasi}]
  Fix any triangle edge $\edge=x_1x_2$ that crosses $\face$ 
  (that is, that intersects the boundary of $\face$ in two points or along an edge).
  The {\em transposal of $\edge$ in $\face$},
  denoted $\transposal{\edge}{\face}$, is defined
  by the following operation:
\begin{window}[0,r,{\includegraphics[height=1.1in]{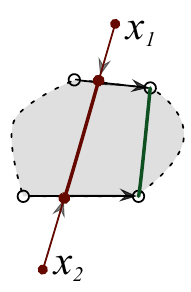}},{}]
  Clip the edge $x_1x_2$ to chord $x_1'x_2' = (x_1 x_2)\cap \face$ of $\face$.
  For each endpoint $x'_i$ of $x'_1x'_2$,
  if the endpoint lies in the interior of an edge $\edge$ of $\face$ 
  (as opposed to being a vertex of $\face$),
  then slide $x'_i$ along $\edge$ to one of the endpoints of $\edge$,
  called the {\em destination} of $x_i$.
  Otherwise (the endpoint is a vertex of $\face$), take that vertex as the destination.
  Choose the destinations (for those where there is a choice) to minimize the length
  of the diagonal that connects the destinations.  (Break ties consistently.)
  The resulting diagonal is $\transposal{\edge}{\face}$.
\end{window}
\end{definition}

Next we define what it means to transpose {\em triangles} and {\em blankets}.
We give a somewhat uninformative formal definition, 
then describe the important properties.
\begin{definition}
  For any triangle $\tri$,
  the {\em transposal of $\tri$ in $\face$},
  denoted $\transposal{\tri}{\face}$,
  is the convex hull of the endpoints of the transposals 
  of the edges of $\tri$ that cross $\face$.

  For any blanket $\blanket\in\blanketSet$,
  the {\em transposal} of $\blanket$ in $\face$,
  denoted $\transposal{\blanket}{\face}$, 
  is the set containing, for each triangle $\tri\in\blanket$, 
  the transposal $\transposal{\tri}{\face}$ of $\tri$.
  That is, $\transposal{\blanket}{\face} = \{ \transposal{\tri}{\face} ~|~\tri\in\blanket\}$.
\end{definition}

\begin{window}[3,r,{\includegraphics[height=0.9in]{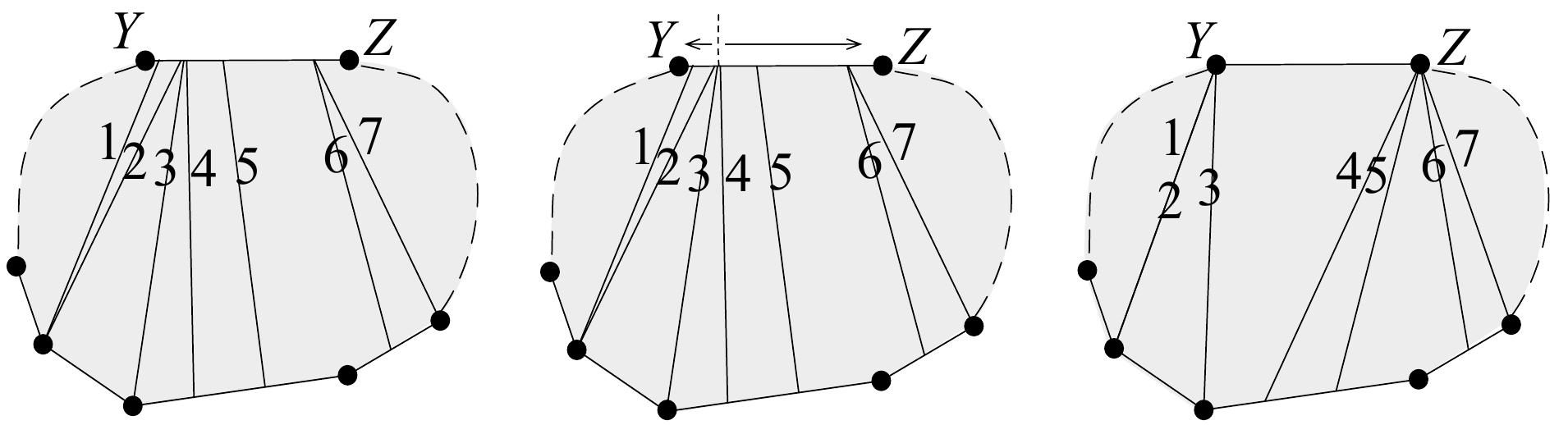}},{}]
  Consider a blanket $\blanket\in\blanketSet$.
  By definition, the edges of triangles in $\blanket$ don't cross within $\face$.
  But, a-priori, their transposals might.
  We next argue that this is not the case.
  In fact, we prove more:
  roughly, that transposing preserves 
  the topology of the partition that $\blanket$ induces on $\face$.
  More precisely, consider that partition, which comes from
  clipping the edges of the triangles in $\blanket$ into $\face$
  as shown to the right.
  Consider any edge $Y Z$.
  Focus on just those chords that have an endpoint in the interior of $Y Z$.  
  Order these chords, as shown in the first of the three pictures,
  according to the order of their endpoints on $Y Z$ going from $Y$ to $Z$.
  For chords sharing an endpoint on $Y Z$,
  break ties in favor of chords that lean closer to $Y$.
\end{window}

\begin{lemma}[transposing preserves order]\label{lemma:ordering}
  In the above ordering of chords along $Y Z$,
  all chords whose endpoints have transposal destination $Y$ 
  precede all chords whose endpoints have destination $Z$.
  (Informally, when transposing the edges, when we slide the endpoints
  to their destinations, the endpoints that slide to $Y$ precede
  the endpoints that slide to $Z$, so no crossings are introduced.)
\end{lemma}

\begin{proof}
Without loss of generality, assume that $Z$ lies (one vertex) clockwise of $Y$.
\begin{window}[2,r,{\includegraphics[width=0.9in]{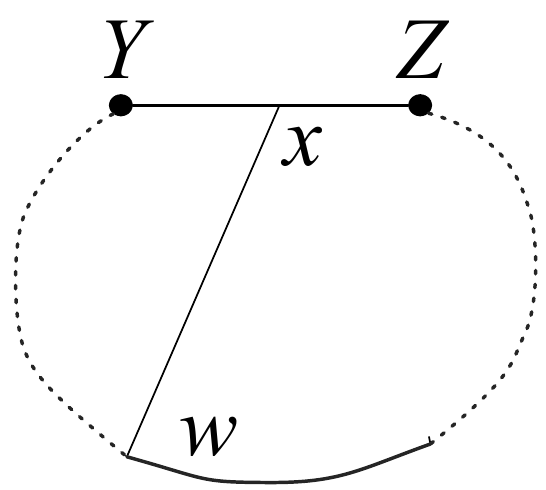}},{}]
\noindent Focus on the chords $xw$ where $x$ is in the interior of $Y Z$.
Let $C_{Y}$ contain those whose endpoint $x$ has destination $Y$.
Let $C_{Z}$ contain those whose endpoint $x$ has destination $Z$.
We show that, if we leave $Y$ and travel {\em counterclockwise} around the boundary to $Z$,
we encounter the chords in $C_Y$ before we encounter the chords in $C_Z$.
This proves the claim, because, as chords in $C$ don't cross, 
as we travel counterclockwise
we must encounter chords in the same order
that we would if traveling clockwise from $Y$ to $Z$.
\end{window}

\begin{window}[0,r,{\smallskip\includegraphics[width=0.9in]{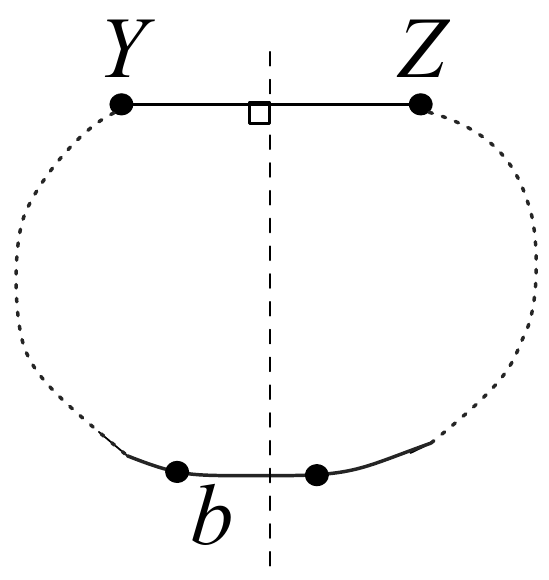}},{}]
Consider the perpendicular bisector of $Y Z$.
Since $\face$ is convex, the bisector intersects the boundary of $\face$
at a single point across from $Y Z$.
Suppose that this intersection point is in the interior of some edge $b$
as shown to the right.  
(If the intersection is a vertex of $\face$, take $b$ to be that point
and follow similar reasoning.)

\hspace*{\parindent}As we travel counterclockwise from $Y$ to $Z$,
until we pass the first endpoint of $b$, every chord endpoint $w$ that we encounter 
is in some edge $\edge$ of $\face$ that lies entirely on the $Y$-side of the bisector.
Since both endpoints of $\edge$ are closer to $Y$ than to $Z$,
no matter which endpoint of $\edge$ is the destination of $w$,
the destination of $x$ will definitely be $Y$.
Thus, until we pass the first endpoint of $b$, we encounter only chords in $C_{Y}$.
\end{window}

As we travel through the interior of edge $b$ (or, if $b$ is a point, through $b$)
for all chord endpoints $w$ that we encounter,
their chords $xw$ will have the {\em same} transposal
(since $x$ is in the interior of $Y Z$ 
and $w$ is in the interior of $b$, and transposing breaks ties consistently).
Thus, traveling through $b$, either we encounter only chords in $C_{Y}$,
or we encounter only chords in $C_{Z}$.

Once we pass reach the other endpoint of $b$, until we reach $Z$,
every chord that we encounter
is an edge of $\face$ that lies entirely on the $Z$-side of the bisector,
so, reasoning as before, we encounter only chords in $C_Z$.
\end{proof}

Because transposing preserves order, 
the topological structure of the {\em transposal} of any blanket $\blanket$
is inherited from the the partition that $\blanket$ induces on $\face$.
Here is an example:

\noindent\includegraphics[width=\textwidth]{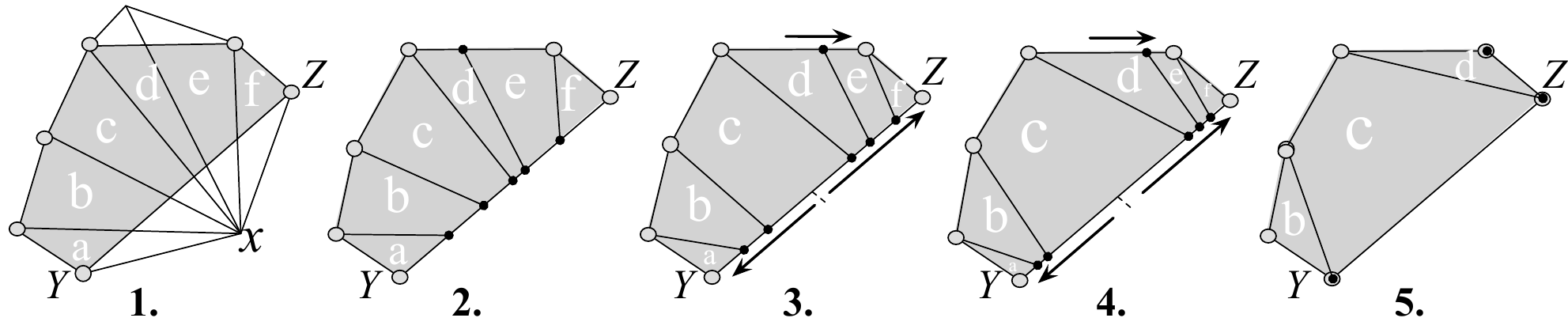}

Above are five copies of a face $\face$ (with gray background).
Copy 1 shows the face blanketed by six triangles.
In copy 2, the triangle edges are clipped to their chords in the face,
giving the partition that $\blanket$ induces on $\face$.
In copies 3 through 5, each chord is shifted to its edge transposal,
by sliding each endpoint to its destination.
Copy 5 shows the resulting edge transposals,
and the transposal of $\blanket$ in $\face$.

Clearly, in the partition that $\blanket$ induces on $\face$ (copy 2)
each region is of the form $\tri\cap\face$ for some $\tri\in\blanket$.
Because transposing preserves order, moving the edges of that partition
to their transposals preserves the topological structure of the partition:
the transposal $\transposal{\blanket}{\face}$ of $\blanket$ (copy 5)
is a convex partition of $\face$
whose edges are the transposals of the edges of $\blanket$,
and whose regions are the transposals of the triangles in $\blanket$.
Also, for each triangle $\tri\in\blanket$,
the boundary of its transposal $\transposal{\tri}{\face}$
consists of the transposals of the edges of $\tri$,
together with up to three edges of $\face$.

The transposal of a blanket is a convex partition of the face,
but not quite a triangulation, because each of its regions may have up to six sides.
To get a triangulation, we simply triangulate each of its regions:

\begin{definition}
The {\em triangulated transposal of a triangle $\tri$ in $\face$},
denoted $\triangulated{\tri}{\face}$,
is the minimum-weight triangulation of the transposal $\transposal{\tri}{\face}$,
except that, if  $\transposal{\tri}{\face}$ has no area, then $\transposal{\tri}{\face}$
is the empty set.
The {\em triangulated transposal of a blanket $\blanket$ in $\face$},
denoted $\triangulated{\blanket}{\face}$,
is the union of the triangulated transposals of the triangles in the blanket.
\end{definition}

The transposal $\transposal{\blanket}{\face}$ is a convex partition of $\face$
whose regions are the transposals of the triangles in $\blanket$,
so the triangulated transposal of $\blanket$ indeed triangulates $\face$.

Finally, we define the fractional triangulation $\transposal{\fracTriang}{\face}$ of $\face$.
We start with the fractional triangulation $\fracTriang$.
We restrict $\fracTriang$ to triangles crossing $\face$.
We decompose this restriction of $\fracTriang$
into a convex combination of blankets of $\face$
(per Lemma~\ref{lemma:blankets}).
Then, in this convex combination, we replace each blanket $\blanket$
by its triangulated transposal $\triangulated{\blanket}{\face}$, a triangulation of $\face$.\footnote
{Alternatively, we could take $\transposal{\fracTriang}{\face}$ to be
the {\em cheapest} triangulation $\triangulated{\blanket}{\face}$ 
over all blankets $\blanket\in\blanketSet$.}
Here is the formal definition:
\begin{definition}
  \label{def:triangulation_transposal}
  Define the {\em transposal of $\fracTriang$ in $\face$},
  denoted $\transposal{\fracTriang}{\face}$, 
  to be the fractional triangulation of $\face$
  formed by the convex combination 
  of the transposals of the blankets in $\blanketSet$, so that
  \[\transposal{\fracTriang}{\face}_{\tri} 
  ~=~ \sum_{\blanket\in\blanketSet} [\tri \in \triangulated{\blanket}{\face}] \,\epsilon_{\blanket}.
  \]
\end{definition}

We now complete Part (i) of the proof of Thm.~\ref{thm:gap}:
\begin{lemma}\label{thm:feasibility}
  Fix any fractional triangulation $\fracTriang$ and any convex face $\face$.
  The transposal $\transposal{\fracTriang}{\face}$ of $\fracTriang$ in $\face$ 
  defined above
  is a feasible fractional triangulation of $\face$.
  That is, it covers the points in $\face$ uniformly with weight 1.
\end{lemma}

\begin{proof}
As discussed, this holds because $\transposal{\fracTriang}{\face}$ is a convex combination
of triangulations of $\face$.
Indeed, it covers each point $p$ in $\face$ with total weight
\[
\sum_{\tri} [p\in t]\transposal{\fracTriang}{\face}_\tri
~=~
\sum_{\tri} \,[p\in t]
\sum_{\blanket\in\blanketSet} [\tri \in \triangulated{\blanket}{\face}] \epsilon_{\blanket}
~=~ \sum_{\blanket\in\blanketSet}\epsilon_{\blanket} \sum_{\tri\in \triangulated{\blanket}{\face}} [p\in \tri]
~=~ \sum_{\blanket\in\blanketSet}\epsilon_{\blanket}
~=~ 1.
\]
The first equality is by definition of $\transposal{\fracTriang}{\face}$.
The second just exchanges the order of summation.
The third holds because $\triangulated{\blanket}{\face}$ triangulates $\face$
(so exactly one $t\in\triangulated{\blanket}{\face}$ contains $p$).
The last follows by Lemma~\ref{lemma:blankets}.
\end{proof}

\subsection{Part (ii) --- bounding the cost}
Fix the convex partition $\convPart$ and fractional solution $\fracTriang$.
By Lemma~\ref{thm:feasibility}, for each face $\face$ of $\convPart$,
the transposal $\transposal{\fracTriang}{\face}$ as defined in Part (i)
is a fractional triangulation of $\face$.
To complete the proof of Thm.~\ref{thm:gap},
we bound the sum of the costs of these fractional triangulations.

We start by observing that we can view $\transposal{\fracTriang}{\face}$
as taking the weight of each triangle $\tri$ in $\fracTriang$,
and transferring that weight to (every triangle in) the triangulated transposal 
$\triangulated{\tri}{\face}$ of $\tri$ in $\face$:
\begin{fact}\label{observation:cost}
  \[\transposal{\fracTriang}{\face}_{t'} = \sum_{\tri\,:\, t'\in \triangulated{\tri}{\face}} \fracTriang_\tri.\]
\end{fact}

\begin{proof} 
\[
\transposal{\fracTriang}{\face}_{\tri'}
~=~
\sum_{\blanket\in\blanketSet} [\tri' \in \triangulated{\blanket}{\face}] \epsilon_{\blanket}
~=~
\sum_{\blanket\in\blanketSet} \sum_{\tri\in\blanket} [\tri' \in \triangulated{\tri}{\face}] \epsilon_{\blanket}
~=~
\sum_{\tri\,:\, \tri'\in \triangulated{\tri}{\face}} 
\sum_{\blanket\in\blanketSet} [\tri'\in\blanket] \epsilon_\blanket
~=~
\sum_{\tri\,:\, t'\in \triangulated{\tri}{\face}} \fracTriang_\tri.
\]
The first equality is the definition of $\transposal{\fracTriang}{\face}$.
The second holds by definition of $\triangulated{\blanket}{\face}$
(namely, $\tri'\in \triangulated{\blanket}{\face}$
iff
$\tri'\in \triangulated{\tri}{\face}$ for some $\tri\in\blanket$).
The third just exchanges the order of summation.
The last follows from Lemma~\ref{lemma:blankets}.
\end{proof}

So far, we've considered how a blanket of triangles transposes into a single $\face$.
Next we consider how a single triangle $\tri$ transposes across multiple faces.
Of course, a given triangle $\tri$ can cross many faces,
but {\em in all but two} its transposal will have no area
(and thus will play no part in the triangulated transposal of $\fracTriang$ in $f$):

\begin{lemma}\label{lemma:transposals}
  Any given triangle $\tri$ crosses at most two faces $\face$ in $\convPart$
  in which its transposal $\transposal{\tri}{\face}$ has positive area.
  Thus, for a given $\tri$, only two faces $\face$ have
  $\cost(\triangulated{\tri}{\face})>0$.
\end{lemma}

\begin{proof}
Fix a triangle $\tri=\Delta{X Y Z}$ and consider how the faces of $\convPart$ can overlap $\tri$.
Say that a face $\face$ is {\em accommodating} if $\tri$'s transposal $\transposal{\tri}{\face}$ in $\face$ has positive area.

\begin{window}[0,r,{\scalebox{.5}{\xfig{transposal-configs}}},{}]
  In the two examples to the right,
  each dashed edge is an edge transposal of an edge of $\tri$.
  Within each accommodating face, the (positive area) transposal of $\tri$ is dark.

  \hspace*{\parindent}%
  We claim that {\em every accommodating face touches all three edges of $\tri$}.
  (A face ``touches'' an edge if the intersection of the face and the edge,
  including boundaries and endpoints, is non-empty.
  For example, the accommodating face 2 on the left of the figure, and 2 and 5 on the right, 
  touch all three edges of $\tri$.  
  Each other face is non-accommodating and, 
  except for 3 and 4 on the right, touches only two edges of $\tri$.)
\end{window}

The claim holds because, if a face $\face$ touches only two edges of $\tri$,
the third edge of $\tri$ lies outside of $\face$,
so the two edges cross the interior of a single edge of $\face$.
Thus, the two edges of $\tri$ that touch $\face$ must have identical transposals,
forcing $\transposal{\tri}{\face}$ to have no area.

Now consider the case that $\tri$ has a face $\face$ that touches the {\em interior}
of all three edges of $\tri$ (as in the figure to the left, above).
Since the faces are non-overlapping and convex,
no face other than $\face$ can touch all three edges of $\tri$.
By the claim, then, only face $\face$ might be accommodating, so the lemma holds.

So assume that no face touches the interiors of all three edges of $\tri$.

By the claim, any accommodating face $\face$ still has to touch all three edges of $\tri$,
but now there is at least one edge, say $X Y$, of $\tri$ whose interior $\face$ avoids.
Thus, $\face$ must touch $X Y$ at an endpoint, say, $Y$.
(For example, consider the figure on the right above.
Faces 2, 3, 4, and 5 touch all three edges of $\tri$, but not all three interiors.)
Since $\face$ touches $X Y$ at $Y$, but does not touch the interior of $X Y$,
there must be an edge $w Y$ of $\face$ that extends through the interior of $\tri$.
Since $w$ is not inside $\tri$, $w Y$ must cut across $\tri$ to the interior of the edge $X Z$.
Thus, 
\begin{equation}\label{eqn:quote}
  {\parbox{0.9\textwidth}{\em in this case, any accommodating face $\face$ must share some vertex $\vertex$ with $\tri$,
and an edge of the face must extend from $\vertex$ across the interior of $\tri$.}}
\end{equation}

If there are two accommodating faces, they must extend an edge across $\tri$ 
from the {\em same} vertex $Y$, for otherwise the extending edges would cross inside $\tri$.

Now consider all edges in $\convPart$ that extend from $Y$ across the interior of $\tri$.
Let these edges be $w_1 Y, w_2 Y, \ldots, w_k Y$, rotating in order around $Y$.
(In the picture above, $k=3$.)
$\convPart$ has $k+1$ corresponding faces $\face_0, \face_1, \ldots, \face_k$, 
also in order rotating around $Y$,
where $\face_{i-1}$ and $\face_i$ share edge $w_i Y$.
By the conclusion (\ref{eqn:quote}) of the paragraph before last,
only these $k+1$ faces might be accommodating.

 To finish, we observe that $\face_i$ is not accommodating
unless $i\in\{0,k\}$ (the first or last face).
Indeed, for $i\not\in\{0,k\}$ edges $w_{i-1} Y$ and $w_i Y$ of $\face_i$ 
extend from $Y$ across $\tri$ to $X Z$.
Since these edges touch at $Y$,
the transposal of $\tri$ in $\face_i$ is thus {\em just the point $Y$}.
Thus, the transposal of $\tri$ in $\face_i$ has no area.
\end{proof}

Our goal is to show that transposing $\fracTriang$ across the faces
increases the cost of $\fracTriang$ by at most a constant factor.
For any triangle $\tri$, 
by the lemma and Fact~\ref{observation:cost}, 
transposing $\fracTriang$
transfers the weight $\fracTriang_\tri$ 
to the triangulated transposals $\triangulated{\tri}{\face}$ of $\tri$
in {\em at most two} faces $\face$.

To proceed we bound the cost of each
$\triangulated{\tri}{\face}$ in terms of the cost of $\tri$.
Recall that $\triangulated{\tri}{\face}$ is the minimum-weight triangulation of
its (non-triangulated) transposal $\transposal{\tri}{\face}$,
which has at most six sides (up to three edges of $\face$,
and up to three transposals of edges of $\tri$).
We start by bounding the cost of $\transposal{\tri}{\face}$.
Our bound depends on the {\em sensitivity} of the edges of the convex partition $\convPart$,
defined as follows:

\begin{definition}[sensitivity]\label{def:sensitivity}
  An edge $\edge$ is {\em $\sensitivity$-sensitive} if, 
  for any potential edge $\edge'$ that crosses $\edge$,
  for each endpoint $x$ of $\edge'$,
  the distance from $x$ to the closest endpoint of $\edge$
  is at most $\sensitivity|\edge'|$.
\end{definition}

(In other words, the circle of radius $\sensitivity|\edge'|$
around each endpoint of $\edge'$ contains an endpoint of $\edge$.)

For the rest of the section, fix $\sensitivity$ 
such that all edges of $\convPart$ are $\sensitivity$-sensitive.

\begin{lemma}\label{lemma:sensitive}
  For any face $\face$ of $\convPart$ and triangle $\tri$,
  the total length of the edges in $\tri$'s transposal $\transposal{\tri}{\face}$
  that are not also edges of $\convPart$
  is at most $2 \sensitivity$ times the length of $\tri$'s edges.
\end{lemma}

\begin{proof}  Let $\face$ be any face of $\convPart$ and $\edge$ be any edge that crosses $\face$. 

    We claim that {\em the length of the edge transposal $\transposal{\edge}{\face}$
      of $\edge$ in $\face$ is at most $2 \sensitivity$ times the length of $\edge$}.
    This claim implies the lemma,
    because each edge of the transposal of $\tri$ (but not of $\convPart$)
    is the edge transposal $\transposal{\edge}{\face}$
    of a unique edge $\edge$ of $\tri$.
    To finish, we prove the claim.
  \begin{window}[0,r,{\scalebox{.5}{\xfig{transposal-length}}},{}]
    For an edge $\edge$ that crosses $\face$, one of the following three cases holds:
    {\bf (1)} $\edge$ is incident to one vertex of $\face$ and properly intersects one $s$ side of $\face$
    (as in the figure immediately to the right),
    {\bf (2)} $\edge$ properly intersects two sides $s$ and $s'$ of $\face$
    (as in the figure to the far right),
    or
    {\bf (3)} $\edge$ is incident to two vertices of $\face$.

    \hspace*{\parindent}In case (1), let $W$ be the vertex that $\face$ shares with $\edge$ (and $\transposal{\edge}{\face}$).
 Since $s$ is $\sensitivity$-sensitive, and $\edge$ crosses $s$,
 the endpoint $W$ of $\edge$ is at most $\sensitivity|\edge|$ 
 from some endpoint of $s$.
 Since $|\transposal{\edge}{\face}|$ is the minimum distance
 between $W$ and any endpoint of $s$, this implies
 $|\transposal{\edge}{\face}| \le \sensitivity|\edge|$.

 \end{window}

 In case (2), let $Y$ be an endpoint of $\edge$ and let $Z$ and $Z'$ respectively be 
 the closest endpoints of $s$ and $s'$ to $Y$.
 Because $\transposal{\edge}{\face}$ is the shortest segment 
 from an endpoint of $s$ to an endpoint of $s'$,
 $|\transposal{\edge}{\face}|\le |Z Z'|$.
 By the triangle inequality,
 \(
 |Z Z'| \le |Y Z| + |Y Z'|.
 \)
 Because $s$ and $s'$ are $\sensitivity$-sensitive,
 $|Y Z|$ and $|Y Z'|$ are each at most $\sensitivity\,|\edge|$.

 In case (3), the transposal $\transposal{\edge}{\face}$ of $\edge$ is the same as $\edge$, 
 so  the claim holds.
\end{proof}

It is straightforward to extend the bound to the {\em triangulated} transposal 
$\triangulated{\tri}{\face}$ of $\tri$.
Recall that the cost of a triangle $t$ is the sum of the costs of its edges,
where the cost of an edge is half its length, unless the edge is on the boundary
of the entire region, in which case the cost of the edge is its length.
The cost $\cost(\triangulated{\tri}{\face})$
of a triangulation $\triangulated{\tri}{\face}$
is the sum of the costs of the triangles 
in the triangulation.

\begin{lemma}\label{lemma:6-gon}
  For any face $\face$ and any triangle $\tri$, 
  the cost $\cost(\triangulated{\tri}{\face})$
  of the triangulated transposal of $\tri$ in $\face$
  is at most three times the cost $\cost(\transposal{\tri}{\face})$ 
  of the (non-triangulated) transposal of $\tri$ in $\face$.
\end{lemma}

\begin{proof}
  Recall that $\transposal{\tri}{\face}$ has at most six vertices,
  say, $v_1,v_2,\ldots,v_k$, ordered clockwise.
  Triangulate $\transposal{\tri}{\face}$ by adding up to three
  interior diagonals connecting the odd vertices
  (e.g.~$v_1 v_3$, $v_3 v_5$, $v_5v_1$).
  The total length of the added diagonals
  is at most the total length of the boundary.
  Likewise, the sum of costs of the added diagonals
  is at most the sum of the costs of the edges on the boundary of $\transposal{\tri}{\face}$.

  Each added edge occurs in two triangles in this triangulation,
  whereas each boundary edge occurs in just one triangle.
  Thus, adding the diagonals gives a triangulation of cost
  at most three times the cost of $\transposal{\tri}{\face}$.
  The lemma follows, as $\cost(\triangulated{\tri}{\face})$ is the minimum cost
  of any triangulation of $\transposal{\tri}{\face}$.
\end{proof}

Next we gather the bounds in the previous lemmas 
to bound the total cost across all the faces.
We are not finished, 
as the bound depends on not only the cost of the fractional triangulation,
but also the total length of the edges in the convex partition $\convPart$
and the sensitivity $\sensitivity$ of those edges:

\begin{lemma}\label{bound}
  The total cost $\sum_\face \cost(\fracTriang^\face)$ is at most
  \(3\|\convPart\|_2 + 12 \sensitivity\, \cost(\fracTriang).\)

\end{lemma}

\begin{proof}
The total cost is
\[
\setlength{\extrarowheight}{6pt}
\begin{array}{r@{~}c@{~~}l@{~~~~}l}
\textstyle \sum_\face \cost(\transposal{\fracTriang}{\face}) 
&=& \textstyle\sum_{\face,\tri} \fracTriang_\tri\, \cost(\triangulated{\tri}{\face})
& \text{by Fact~\ref{observation:cost}}
\\
&\le& 3\sum_{\face,\tri} \,
[\cost(\triangulated{\tri}{\face})>0]~
\fracTriang_\tri \,
\cost(\transposal{\tri}{\face})
& \text{by Lemma~\ref{lemma:6-gon}}
\\
&\le& 3\|\convPart\|_2 + ~\,6\sensitivity\, \sum_{\tri,\face} ~
[\cost(\triangulated{\tri}{\face})>0]~
\fracTriang_\tri \,\cost(\tri)
& \text{by Lemma~\ref{lemma:sensitive}}
\\
&\le& 3\|\convPart\|_2 + 12 \sensitivity \,\sum_{\tri}\, \fracTriang_\tri\, \cost(\tri)
& \text{by Lemma~\ref{lemma:transposals}}
\\ 
&=&3\|\convPart\|_2 + 12 \sensitivity\, \cost(\fracTriang)
& \text{by definition of $\cost(\fracTriang)$.}
\end{array}
\]
\end{proof}

To proceed further,
we need a convex partition whose edges have constant sensitivity
and total length $O(\cost(\fracTriang))$.
Levcopoulos and Krznaric have shown the existence 
of something close: a convex partition $\LK$ whose edges are 4.45-sensitive
and have total length $O(\MCP(\graph))$
(recall that $\MCP(\graph)$ is the minimum-length convex partition of $\graph$):

\begin{lemma}[\cite{krznaric1998quasi}]\label{lemma:LK}
  For some constant $\lambda>0$, for any MWT instance $\graph$,
  there exists a convex partition $\LK$ of $\graph$,
  whose edges are 4.45-sensitive,
  having total length $\|\LK\|_2 \le \lambda\,\|\MCP(\graph)\|_2$.
\end{lemma}

\begin{proof}
  Levcopoulos and Krznaric 
  show that what they call the {\em quasi-greedy convex partition}
  has these properties:
  for Property (1), see their Lemma 5.4 and the discussion before it;
  for Property (2), see their Corollary 5.3   \cite{krznaric1998quasi}.
\end{proof}

This convex partition will work for us:
we prove next that $\|\MCP(\graph)\|_2 \le 18\,\cost(\fracTriang)$.

(Note that $\|\MCP(\graph)\|_2$ is trivially at most the cost of any {\em integer}
triangulation, but the bound here concerns the {\em fractional} triangulation,
so requires proof.)

The proof uses the constraints on $\fracTriang$ and leverages
a previous analysis of $\MCP(\graph)$
due to Plaisted and Hong \cite[Lemma 10]{plaisted1987heuristic}.
\begin{lemma}\label{lemma:fracmcp}
 $\|\MCP(\graph)\|_2\le 18\, \cost(\fracTriang)$
\end{lemma}

\begin{proof}
  For every vertex $\vertex$ in the interior of the convex hull of the vertex set $\vertices$, 
  define a {\em star} at $\vertex$ to be a subset of edges incident to $\vertex$
  such that no two successive edges (around $\vertex$) 
  are separated by an angle of 180 degrees or more.
  For every vertex $\vertex$ on the boundary of the convex hull of $\vertices$, 
  define the (only) star at $\vertex$ to consist
  of the two boundary edges  incident to $\vertex$.
  Let $S_{\min}(\vertex)$ denote the minimum cost of any star at $\vertex$.
  Plaisted and Hong show $\|\MCP(\graph)\|_2 \le 6 \sum_{\vertex\in\vertices} S_{\min}(\vertex)$
  \cite[Lemma 10]{plaisted1987heuristic}.

  We claim $S_{\min}(\vertex) \le (3/2) \sum_{\edge \ni \vertex} \fracTriang_\edge|\edge|$, where $\fracTriang_\edge$ is $\sum_{\tri\ni\edge}\fracTriang_\tri$ if $\edge$ is on the boundary of the convex hull, 
  and otherwise half this amount.
  As $\sum_\vertex \sum_{\edge\ni \vertex} \fracTriang_\edge|\edge| = 2\sum_\edge \fracTriang_\edge|\edge| = 2\,\cost(\fracTriang)$,
  the claim implies the lemma.
  We prove the claim.

  It's easy to see that, for any boundary vertex $\vertex$, $S_{\min}(\vertex) \le \sum_{\edge \ni \vertex} \fracTriang_\edge|\edge|$, 
  so restrict attention to just an interior vertex $\vertex$ and its edges.

  \newcommand{\wrap}{\text{wrap}}

  \begin{window}[0,r,{\includegraphics[width=1.3in]{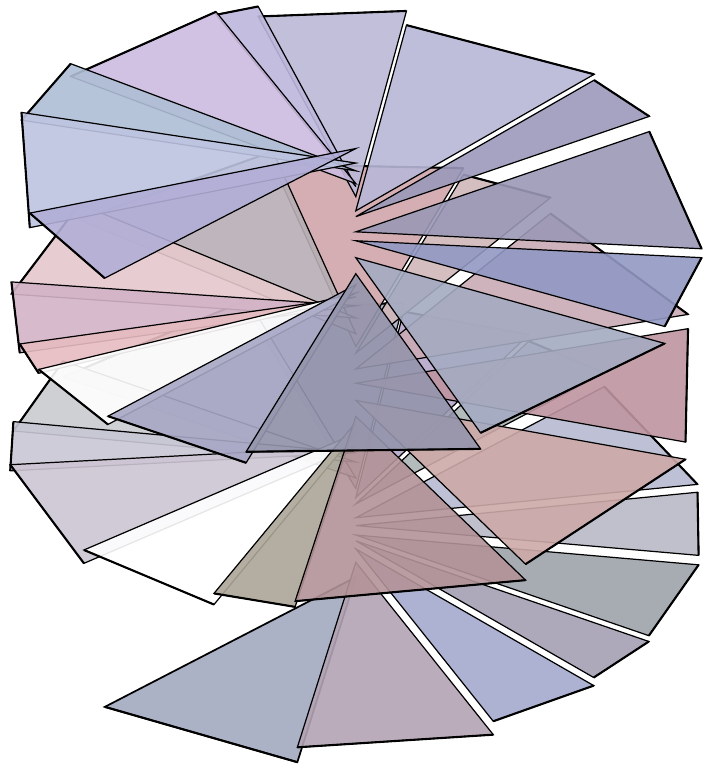}},{}]
    Because $\fracTriang$ satisfies Constraint (\ref{constraint:edge}),
    rotating clockwise around $\vertex$, 
    there is a sequence $\vertex_1,\vertex_2,\ldots,\vertex_k$
    of distinct vertices
    such that for each $i=1,\ldots,k$,
    the triangle $\vertex\, \vertex_i\, \vertex_{(i+1)\bmod k}$
    has positive weight in $\fracTriang$.
    (To find the sequence, take any positive-weight triangle that has $v$ as a vertex.
    Let $\vertex_1$ and $\vertex_2$ be the other vertices, in clockwise order.
    By Constraint (\ref{constraint:edge}), there is a positive-weight triangle 
    that shares edge $\vertex \vertex_2$ and lies clockwise of that edge.
    Let $\vertex_3$ be the other vertex of that triangle.
    By Constraint (\ref{constraint:edge}), there is a positive-weight triangle 
    that shares edge $\vertex \vertex_3$ and lies clockwise of that edge.
    Let  $\vertex_4$ be the other vertex of that triangle.
    Continue, stopping when the next vertex $\vertex_i$
    that would be added is $\vertex_1$ ---
    this must happen by Constraint (\ref{constraint:edge}).)
  \end{window}
  
  Let $\edge_i$ and $\tri_i$ denote edge $\vertex \vertex_i$ 
  and triangle $\vertex\, \vertex_i\, \vertex_{(i+1)\bmod k}$, respectively.
  Note that each edge, and each triangle, is distinct.
  Call the sequence of edges $h = (\edge_1,\edge_2,\ldots,\edge_k)$ a {\em helix}.
  Let $\wrap(h)$ denote the number of times $h$ wraps around $\vertex$.
  By a standard construction
  the $\fracTriang_\edge$'s 
  can be expressed as a linear combination of incidence vectors of helices.
  (Similar to Lemma~\ref{lemma:blankets}'s proof,
  repeatedly find a helix $h$, choose weight $\epsilon_h$,
  and subtract $\epsilon_h/\wrap(h)$ from each triangle $\fracTriang_{\tri_i}$ in the helix,
  reducing coverage near $\vertex$ by $\epsilon_h$.)
  This gives a probability distribution $\epsilon$ on helices $h$
  such that each $\fracTriang_\edge = \sum_{h} [\edge\in h] \epsilon_h/ \wrap(h)$.

  Now choose a helix $h$ at random from the probability distribution $\epsilon$.
  Partition $h$ greedily into contiguous subsequences of edges such that each
  group $g$ is maximal subject to the constraint that the total clockwise angle around $v$
  swept by the group's edges is at most $360^\circ$.
  (In the figure, white triangles separate groups.)
  Consideration shows that each group contains a star,
  and (as neighboring groups are separated by at most $180^\circ$),
  there are at least $\lceil 360\, \wrap(h) /(360+180)\rceil  
  = \lceil 2\, \wrap(h)/3\rceil$ groups.

  From the randomly chosen $h$, choose a group $g$ uniformly at random 
  from $h$'s first $\lceil 2 \wrap(h)/3\rceil$ groups.
  For any given edge $\edge$, the probability that $\edge$ is in $g$ is at most
  $\sum_{h} \,[\edge \in h] \epsilon_h/ (2\,\wrap(h)/3) = (3/2) \fracTriang_\edge$.
  Thus, by linearity of expectation, the expected total length $E[\|g\|_2]$ of edges in $g$ 
  is at most $(3/2) \sum_{\edge\ni \vertex} \fracTriang_\edge|\edge|$.
  On the other hand, every $g$ contains a star, so $S_{\min}(\vertex)  \le E[\|g\|_2]$.
  This proves Lemma~\ref{lemma:fracmcp}.
\end{proof}

For the rounding procedure,
fix the (previously arbitrary) convex partition $\convPart$ 
to be the partition $\LK$ from Lemma~\ref{lemma:LK}.
The cost of the final triangulation is at most
\[
\setlength{\extrarowheight}{4pt}
\begin{array}{@{~}c@{~~}ll}
  \lefteqn{\textstyle \sum_\face \cost(\transposal{\fracTriang}{\face})} && \text{Because each face $\face$ is simple.}
\\\le & 3\|\LK\|_2 + 12\,\sensitivity\, \cost(\fracTriang)  & \text{By Lemma~\ref{bound}.}
\\\le & 3\lambda \|\MCP(\graph)\|_2 + 54\, \cost(\fracTriang) & \text{By Lemma~\ref{lemma:LK}.}
\\\le & 3\lambda \cdot 18\, \cost(\fracTriang) + 54\, \cost(\fracTriang) & \text{By Lemma~\ref{lemma:fracmcp}.}
\\= & 54(\lambda+1)\, \cost(\fracTriang) &
\end{array}
\]
Hence, the integrality gap is at most $54(\lambda+1)$,
completing the proof of Thm.~\ref{thm:gap}.
\qed



\section{LP generalizes heuristics}\label{sec:heuristics}

This section proves our second new result 
(the LP generalizes MWT heuristics).
Here is a summary of heuristics for determining that 
a given potential edge $\edge = xy$ of $\graph$ is in every MWT 
of a given MWT instance $\graph=(\vertices,\edges)$:

\newcommand{\thing}[2]{\smallskip \par\item[#1]{#2}}

\begin{description}
  \thing{$\beta$-skeleton:}
  {
    For $\beta = 1/\sin k $ where $k \approx \pi/3.1$,
    there does not exist a point $z$
    in the point set such that $\angle x z y \ge \arcsin(1/\beta) \approx \pi/3.1$.
    Equivalently, 
    the two disks of diameter $\beta\,|\edge|$ having $\edge$ as a chord are empty
    of points.
    If this condition holds, then $e$ is in every MWT of $G$
    \cite{keil1994computing,cheng1996approaching}.
  }
  \thing{$\YXY$-subgraph:}
  {
    For every potential edge $\overline{p q}$ that crosses $\edge=\overline{x y}$,
    its size $|\edge|$ is at most $\min\{|\overline{p x}|,|\overline{p y}|,|\overline{q x}|,|\overline{q y}|\}$.
    If this condition holds, then $e$ is in every MWT of $G$
    \cite{yang1994chain,gilbert1979new}.
  }
  \thing{maximality:}
  {
    For every potential edge that crosses $\edge$, that edge is known to be {\em excluded from} every MWT.
    If so, then $e$ is in every MWT of $G$
    (see e.g.~\cite{dickerson1997large}).
  }
\end{description}
\smallskip

Here is a summary of heuristics for determining that 
a given potential edge $\edge$ of $\graph$ 
(not on the boundary of the region to be triangulated) 
is excluded from every MWT of $\graph$:

\begin{description}
  \thing{independence:}
  {
  Some potential edge that crosses $\edge$ is known to be in every MWT.
  If this condition holds, then $e$ is not in any MWT of $G$
  (see e.g.~\cite{dickerson1997large}).
  }
  \thing{diamond:}
  {
  Neither of the two triangles with base $\edge$ and base angle $\pi/4.6$ are empty.
  If this condition holds, then $e$ is not in any MWT of $G$
  \cite{das1989triangulations,drysdale2001exclusion}.
  }
  \thing{LMT skeleton:}
  {
    For every two triangles $\tri$ and $\tri'$ for which $\edge$ is {\em locally minimal},
    one of the edges of $\tri$ or $\tri'$ is known to be excluded from every MWT.
    If this condition holds, then $e$ is not in any MWT of $G$
    \cite{dickerson1997large}.
    (Edge $\edge$ is {\em locally minimal} 
    for two triangles $\tri$ and $\tri'$
    if $\tri\cap \tri' = \edge$ and
    $\tri$ and $\tri'$ together are a minimum-length triangulation of 
    the quadrilateral $Q = \tri\cup \tri'$ ---
    that is, either $Q$ is non-convex,
    or $\edge$ is no longer than the other diagonal of $Q$.)
  }
\end{description}
\smallskip

\noindent
Let $\edges^*$ denote the set of edges that can be deduced to be in every MWT
by applying the logical closure of the above six rules.
(Logical closure is necessary because the maximality, independence, and LMT-skeleton conditions
depend on the known statuses of edges other than $\edge$.
For example, if one of the conditions implies that some edge $\edge'$ is excluded from every MWT,
then the LMT-skeleton condition may then in turn 
imply that some new edge $\edge$ is excluded from every MWT,
because $\edge'$ lies on one of two triangles $\tri$ or $\tri'$ in the pair for which $\edge$ is locally minimal.)

Ideally, the set $\edges^*$ gives a partition of $\graph$ in which every face is empty.
If this happens, then the remaining edges in the MWT can be found
by triangulating each remaining face independently using the standard
dynamic-programming algorithm, and we say $\graph$ is {\em solvable} via the heuristics.
According to \cite{dickerson1997large} (1997),
most random instances with as many as 40,000 points are solvable via the heuristics.\footnote
{\cite{dickerson1997large} define the modified LMT-skeleton to be
the set of edges that can be deduced to be in every MWT via 
(the logical closure of) just the 
diamond, LMT-skeleton, maximality, and independence conditions above.
The use of logical closure is crucial to the effectiveness of the LMT skeleton.
}  

Here is our second new result.  If an instance is solvable via the heuristics,
then LP \refLP solves the instance also:
\begin{theorem}\label{thm:heuristics}
  For any instance $\graph$ of MWT,
  let $\edges^*$ be the partition of $\graph$ defined above.
  If every face of $\edges^*$ is empty, then every optimal extreme point of the LP (for $\graph$)
  is the incidence vector of a minimum-length triangulation.
\end{theorem}

The remainder of the section gives the proof. 
The first step is to show that each condition above that ensures that an edge is in
(or excluded from) every MWT also ensures that the LP gives the edge weight 1 (or 0)
in any optimal fractional solution.

Say that LP~\refLP\ {\em forces a potential edge $\edge$ to $z$} (where $z\in\{0,1\}$)
if, for every optimal fractional triangulation $\fracTriang^*$ of $\graph$,
the weight that $\fracTriang^*$ gives to $\edge$ is $z$.

\begin{lemma}\label{lemma:in}
  If any of the following conditions holds, the LP forces potential edge $\edge$ of $\graph$ to 1.
  \begin{enumerate}
  \item \label{in:beta}
    {\bf $\beta$-skeleton:} The $\beta$-skeleton condition above holds for $\edge$.
  \item \label{in:YXY}
    {\bf $\YXY$-subgraph:} The $\YXY$-subgraph condition above holds for $\edge$.
 \item \label{in:maximality}
    {\bf maximality:}
    The LP forces every potential edge that crosses $\edge$ to 0.
  \end{enumerate}
\end{lemma}

\begin{proofidea}
  Part~(\ref{in:maximality}) is relatively straightforward:
  if $\fracTriang^*$ gives weight 0 to every edge that crosses $\edge$,
  then no triangle $\tri$ that crosses $\edge$ has positive $\fracTriang^*_\tri$, 
  so the only way $\fracTriang^*$ can cover points near $\edge$ is with triangles that have $\edge$ as a side.

  The original $\beta$-skeleton and the $\YXY$-subgraph heuristics 
  are shown to be valid for MWT by local-improvement arguments:
  if the condition holds for an edge $\edge$ that is {\em not} in the MWT,
 then a polygon $\polygon$ covering $\edge$ within the MWT can be retriangulated at lesser cost,
  contradicting the optimality of the MWT
  \cite{keil1994computing,cheng1996approaching,yang1994chain,gilbert1979new}.
  Here we extend those arguments to any optimal {\em fractional} triangulation $\fracTriang^*$:
  if the condition holds and $\fracTriang^*$ does not give $\edge$ weight 1,
  then a polygon $\polygon'$ covering $\edge$ whose triangles have positive weight in $\fracTriang^*$
  can be retriangulated (lowering the weight of those triangles by $\epsilon>0$
  and raising the weight of other triangles by $\epsilon$),
  giving a fractional triangulation that costs less than $\fracTriang^*$.

  The original arguments are intricate geometric case analyses, typically taking several
  pages.  The arguments do not extend completely to our setting for the following reason:
  in the MWT setting, the polygon $\polygon$ identified for re-triangulation 
  is the union of non-crossing triangles,
  whereas here, in the fractional setting, 
  the polygon $\polygon'$ is the union of triangles that {\em can} cross
  (much as in Lemma~\ref{lemma:blankets}).
  If the triangles in $\polygon'$ do not cross, then the original arguments apply,
  but in general additional analysis is needed.
  \begin{window}[0,r,{\resizebox{2.5in}{1.25in}{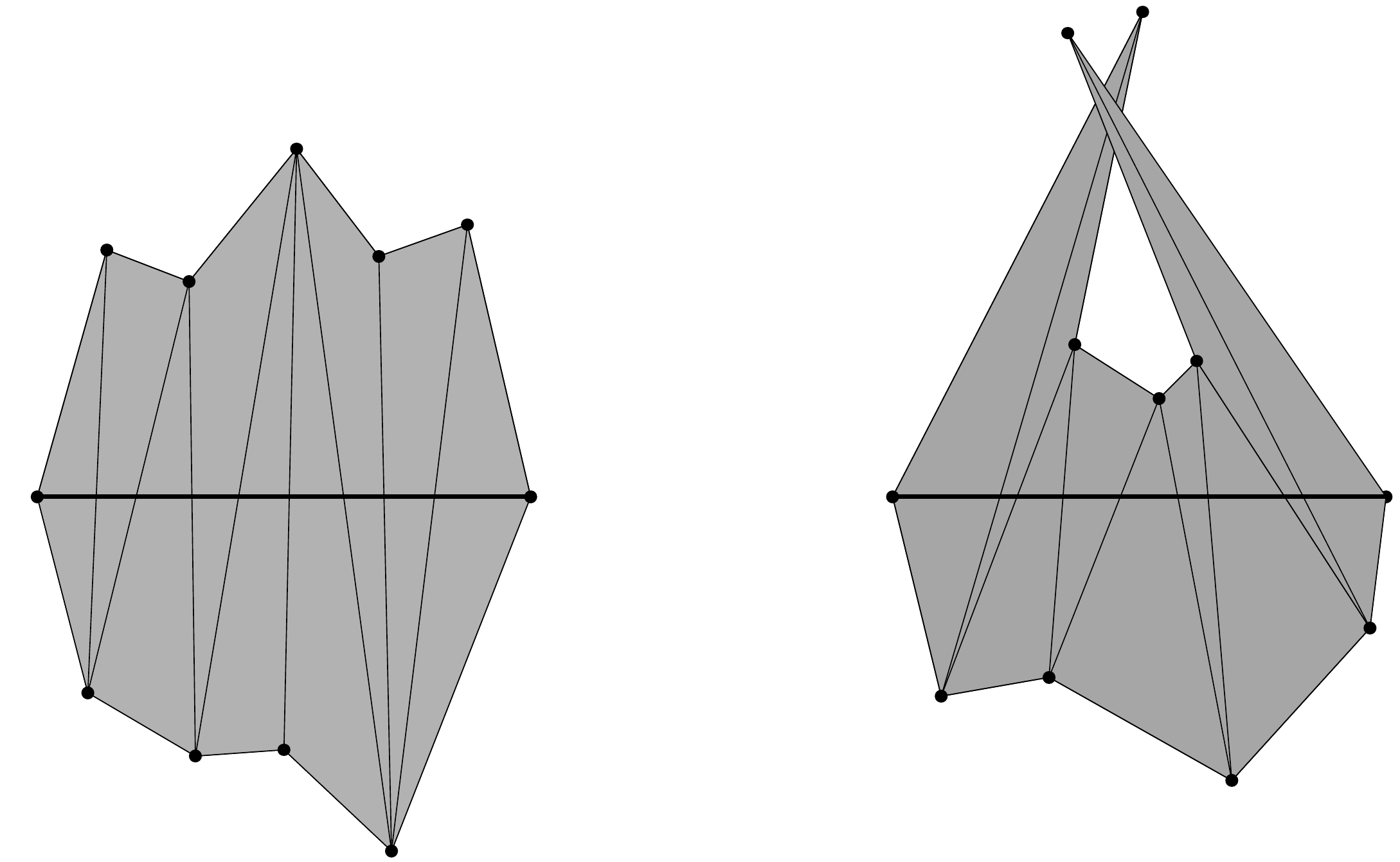}},{}]
    To illustrate, consider the $\beta$-skeleton.
    Suppose for contradiction that the $\beta$-skeleton condition
    holds for an edge $\edge=\overline{a b}$ but $e$ does not occur in the MWT.
    The previous work \cite{keil1994computing,cheng1996approaching} 
    shows that there must be
    a sequence $\tri_1,\tri_2,\ldots,\tri_m$ of empty triangles in the MWT whose union $\polygon$
    covers $\edge$ as shown in the left of the figure to the right.
    Using the $\beta$-skeleton condition, they show that this union 
    has a triangulation that costs less than does $\tri_1,\ldots,\tri_m$, 
    contradicting the optimality of the MWT.

    \hspace*{\parindent}
    In the current context, if $\edge$ has weight below 1 in $\fracTriang^*$,
    then there must (similarly) exist a sequence $\tri_1,\tri_2,\ldots,\tri_m$ of empty triangles
    with positive weight in $\fracTriang^*$ covering $\edge$, 
    but these triangles can {\em cross} (see the example to the right above).
    We extend their arguments to show that, even if such
    crossing occurs, a triangulation of lower cost can still be found.%
  \end{window}%
\end{proofidea}


\begin{fullproof}
Here are the details of the proofs for Part \ref{in:beta} ($\beta$-skeleton) and Part \ref{in:YXY} ($\YXY$-subgraph).
Part \ref{in:maximality} (maximality) is discussed already above, in the proof idea.
%


\medskip
\noindent
{\bf Part \ref{in:beta} ($\beta$-skeleton):}
The original $\beta$-skeleton heuristics 
are shown to be valid for MWT by local-improvement arguments:
if an edge $\edge$ is in the $\beta$-skeleton (for $\beta \approx 1/\sin(\pi/3.1)$) but {\em not} in the MWT,
then a polygon $\polygon$ covering $\edge$ within the MWT can be retriangulated at lesser cost,
contradicting the optimality of the MWT \cite{keil1994computing,cheng1996approaching}.
We briefly sketch their argument and then extend it to any optimal {\em fractional} triangulation $\fracTriang^*$.

Assume for the remainder of this section that $\edge$ goes horizontally from 
the point $a$ on the left to the point $b$ on the right.
If $a b$ is not in the MWT, there exists a set of MWT edges that intersect $\edge$. 
Let $\edge_1,\ldots,\edge_n$, be the set of edges indexed in  
non-decreasing order of their length. 
If the edges are removed from the MWT, an empty polygonal region $\polygon$ results. 
In \cite{keil1994computing,cheng1996approaching}
it is shown that 
$\polygon$ can be retriangulated at lesser cost by a set of edges that contains $a b$. 
The idea is to generate a sequence of triangulated 
polygons $\polygon_0,\ldots,\polygon_n$ such that $\polygon_0$ is the degenerate polygon $a b$, $\polygon_n$ is 
a triangulation of $\polygon$ and $\polygon_{j-1}\subseteq \polygon_j$. To obtain $\polygon_j$, $\polygon_{j-1}$ is expanded 
to include the endpoints $\vertex_j$ and $\vertex'_j$ of $\edge_j$. Assume $\vertex_j$ is above the line through 
$a b$ and $\vertex'_j$ is below it. 
If both $\vertex_j$ and $\vertex'_j$ already lie on the boundary of $\polygon_{j-1}$ then $\polygon_j=\polygon_{j-1}$. 
Otherwise, at least one of them will not be on the boundary of $\polygon_{j-1}$.
Assume without loss of generality $\vertex_j$ is not on the boundary of $\polygon_{j-1}$ 
(If $\vertex'_j$ is also outside $\polygon_{j-1}$, it will be dealt with similarly).

Since $\vertex_j$ is not on the boundary of $\polygon_{j-1}$, edge $\edge_j$ intersects a boundary edge $\vertex_i \vertex_k$ of $\polygon_{j-1}$.
Consider the sequence $\delta$ of vertices on the path from $a$ to $b$ on the boundary of $\polygon$
(there are two such paths, but the one above the line through $a b$ is intended).
On the sequence $\delta$, vertex $\vertex_i$ is the last vertex before $\vertex_j$ that belongs to $\polygon_{j-1}$ and 
$\vertex_k$ is the first vertex after $\vertex_j$ that belongs to $\polygon_{j-1}$. 
This observation allows us to clearly define $\vertex_i \vertex_k$ in 
the fractional setting because in that setting, polygon $\polygon$ may be self-intersecting 
and $\edge_j$ may intersect more than one boundary edge of $\polygon_{j-1}$ in the half-space above $a b$.
In general, the triangle $\triangle{\vertex_i \vertex_j \vertex_k}$ contains a
subsequence $\delta_1$ of vertices on $\delta$ from $\vertex_i$ to $\vertex_j$ 
and another subsequence $\delta_2$ from $\vertex_j$ to $\vertex_k$. 
The polygon $\vertex_i\delta_1 \vertex_j\delta_2 \vertex_k$ is then triangulated arbitrarily, 
and $\polygon_j$ is the union of $\polygon_{j-1}$ and the triangulated polygon $\vertex_i\delta_1 \vertex_j\delta_2 \vertex_k$ and 
possibly another triangulated polygon to include $\vertex'_j$ if $\vertex'_j$ is not on the boundary of $\polygon_{j-1}$.
\begin{window}[0,r,{\resizebox{2in}{1.8in}{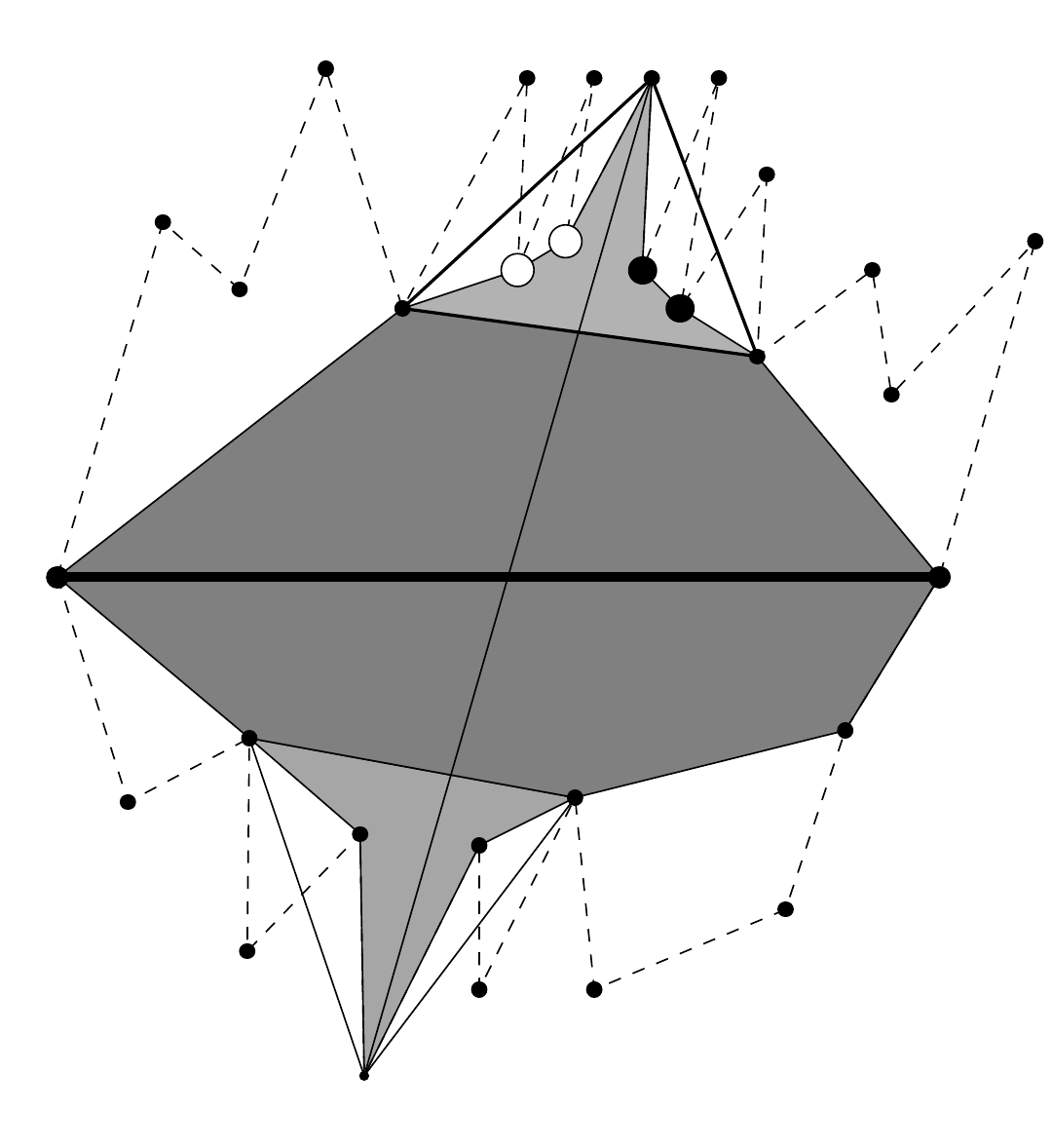}},{}]
This construction is shown in the figure to the right.
The polygon with dashed boundary is $\polygon$, and $\polygon_{j-1}$ is a triangulation of the dark gray polygon.
The union of $\polygon_{j-1}$ and arbitrary triangulations of the light gray polygons is $\polygon_j$ 
which includes the endpoints $\vertex_j$ and $\vertex'_j$ of $\edge_j$. 
The light gray polygon above $a b$ is $\vertex_i\delta_1 \vertex_j\delta_2 \vertex_k$.
The white vertices inside triangle $\triangle{\vertex_i \vertex_j \vertex_k}$ are
$\delta_1$, and the black vertices inside the triangle are $\delta_2$.

\hspace*{\parindent}
In \cite{keil1994computing,cheng1996approaching} it is shown by induction on $j$ that 
for every $j$ all edges in the triangulation of $\polygon_j$ are shorter than $\edge_j$.
By induction all the edges of $\polygon_{j-1}$ are shorter than $\edge_{j-1}$ 
and $|\edge_{j-1}|\le |\edge_j|$. Thus it remains to show that all the new edges that are added to $\polygon_{j-1}$ to 
form $\polygon_j$ are shorter than $\edge_j$. The new edges triangulate the polygon $\vertex_i\delta_1 \vertex_j\delta_2 \vertex_k$ 
and they are all contained in triangle $\triangle \vertex_i \vertex_j \vertex_k$, so each new edge is shorter than $\max\{|\vertex_i \vertex_j|,|\vertex_j \vertex_k|,|\vertex_i \vertex_k|\}$.
Since $\vertex_i \vertex_k\in \polygon_{j-1}$, $|\vertex_i \vertex_k| < |\edge_{j-1}|\le |\edge_j|$. It remains to show that $|\vertex_i \vertex_j|< |\edge_j|$.                                     
The argument for $\vertex_j \vertex_k$ is similar.
The following two facts are used for this part of the argument.  
\end{window}

\begin{fact}[{\cite[Lemma 2]{keil1994computing}}]\label{prop:length}
Let $a b$ be an edge in the $\beta$-skeleton (where $\beta \approx 1/\sin(\pi/3.1)$). 
For any edge $p q$, if $p q$ intersects $a b$, then it has length $|p q|$ greater than 
$\max\{|a b|,|a p|,|a q|,|b p|,|b q|\}$.
\end{fact}

\begin{fact}[{\cite[Remote Length Lemma]{cheng1996approaching}}]
\label{prop:remote}
Let $a b$ be an edge in the $\beta$-skeleton (where $\beta \approx 1/\sin(\pi/3.1)$).
Let $p,q,r$ and $s$ be four other distinct points of the point set 
such that $p q$ intersects $a b$, $r s$ intersects $a b$, $p q$ 
does not intersect $r s$, and $p$ and $s$ lie on the same side of the line through $a b$. 
Then $|q r|< \max\{|p q|,|r s|\}$.
\end{fact}

The argument to show $|\vertex_i \vertex_j|< |\edge_j|$ is as follows.
If $\vertex_i$ lies in triangle $\triangle a \vertex_j b$, then
 $|\vertex_i \vertex_j|\le \max\{|av_j|,|\vertex_j b|,|a b|\}< |\edge_j|$.
The second inequality holds based on Fact~\ref{prop:length}.
If $\vertex_i$ is outside $\triangle a \vertex_j b$, consider the convex hull of the path from 
$a$ to $\vertex_j$ on $\polygon_j$. 
Vertex $\vertex_i$ must lie in some triangle $\triangle \vertex_c \vertex_d \vertex_j$ where $\vertex_c$ and $\vertex_d$ are hull vertices. Thus, 
$|\vertex_i \vertex_j|\le \max\{|\vertex_c \vertex_j|,|\vertex_c \vertex_d|,|\vertex_d \vertex_j|\}$. Since $\vertex_c$ and $\vertex_d$ are hull vertices, 
they were added in the growth process in the past. Thus, the edges $\edge_c$ and $\edge_d$ with endpoints 
$\vertex_c$ and $\vertex_d$ were processed before $\edge_j$, so neither $\edge_c$ nor $\edge_d$ is longer
than $\edge_j$. Combining this observation with Fact~\ref{prop:remote} gives $|\vertex_c \vertex_j| < \max\{|\edge_c|,|\edge_j|\}\le|\edge_j|$. 
Using a similar argument, one can show that $\vertex_c \vertex_d$ and $\vertex_d \vertex_j$ are both shorter than $\edge_j$. 
Thus, $|\vertex_i \vertex_j|\le\max\{|\vertex_c \vertex_j|,|\vertex_c \vertex_d|,|\vertex_d \vertex_j|\} <|\edge_j|$. This completes the proof that 
all the edges of $\polygon_j$ are shorter than $\edge_j$, and thus the new triangulation of $\polygon$ costs less.

Next we extend the above arguments to any optimal {\em fractional} triangulation $\fracTriang^*$. 
If $\fracTriang^*$ does not give $a b$ weight 1,
there is a triangle $\tri$ with positive $\fracTriang^*_t$ that properly intersects $a b$.
For any side $d$ of $\tri$ that intersects $a b$ there must be a triangle $\tri'$ 
with positive $\fracTriang^*_{\tri'}$ that has $d$ as a side and lies on the other side of $d$ from $\tri$. 
The existence of a triangle $\tri'$ is a consequence of Constraints (\ref{constraint:edge}). 
Repeating the same argument for the new triangle(s) gives a set of triangles that cover $a b$. 

\begin{window}[0,r,{\resizebox{3.5in}{1.7in}{\input{./pictures/tsequence.pdf_t}}},{}]
Let $\Gamma = (t_1,t_2,\ldots,t_m)$ be the sequence of triangles in the order they intersect 
$a b$ in the direction from $a$ to $b$. Triangle $t_1$ is incident on $a$ and $t_m$ is incident on $b$. 
All triangles have a positive weight in $\fracTriang^*$.
The triangles in the sequence may or may not cross each other as shown in the figure to the right.
Let $\polygon$ (the shaded area in the figure) be the polygon formed by the boundary edges of triangles in the sequence.
As shown in the right figure polygon $\polygon$ is self-intersecting if some of the triangles in the sequence cross.
Next, we consider both cases and derive a contradiction in each case.
\end{window}


\noindent{\em Case 1} --- no triangles in $\Gamma$ cross.  For this case we apply directly the technique used in 
\cite{keil1994computing} and \cite{cheng1996approaching} to retriangulate the interior of $\polygon$ at lower cost.
Lowering the weight of those triangles in $\Gamma$ by $\epsilon>0$
and raising the weight of new triangles by $\epsilon$,
gives a fractional triangulation of cost less than $\cost(\fracTriang^*)$.


\noindent
{\em Case 2 ---} some triangles in $\Gamma$ cross.  In this case the technique of \cite{keil1994computing} and 
\cite{cheng1996approaching} cannot be directly applied because,
in that setting, the polygon $\polygon$ identified for retriangulation 
is the union of non-crossing triangles,
whereas in this case, $\polygon$ is the union of triangles that cross.

Let $d_1,\ldots,d_n$ be the set of edges of triangles in $\Gamma$ that intersect $a b$ 
indexed in the order they intersect $a b$ (in the direction from $a$ to $b$).
The only part of the argument used in \cite{keil1994computing} and 
\cite{cheng1996approaching} that doesn't go through concerns Fact~\ref{prop:remote}. 
Fact~\ref{prop:remote} holds for any pair of edges $p q$ and $r s$ that 
intersect $a b$ and do not intersect each other. Recall that in the MWT setting the edges intersecting $a b$ 
do not intersect each other, so Fact~\ref{prop:remote} holds for any pair 
of those edges. However, in the current case some edges in $d_1,\ldots,d_n$ may intersect each other. 
Thus, Fact~\ref{prop:remote} does not automatically hold in this case.  
This issue is resolved by the following technical lemma.
\begin{fact}\label{obs:point}
Let $d_i=p q$ and $d_j=r s$ be two edges of triangles $t_1,t_2,\ldots,t_m$ such that 
$p q$ intersects $a b$, $r s$ intersects $a b$, $p q$ 
intersects $r s$, and $p$ and $s$ lie on the same side of the line through $a b$.
Then $p s$ and $q r$ are both shorter than $\max\{|p q|,|r s|\}$.  
\end{fact}


\begin{proof}
Assume that endpoints $p$ and $s$ are above the line through $a b$.
Let $\ell$ be the intersection point of $d_i$ and $d_j$, and assume without loss of generality that 
point $\ell$ is also above the line through $a b$.
Let $a_i$ and $a_j$ be the intersection points of $d_i$ and $d_j$ with $a b$.
Assume $i < j$ so that $|a a_i|< |a a_j|$. 
\begin{window}[2,r,{\resizebox{2.3in}{1.8in}{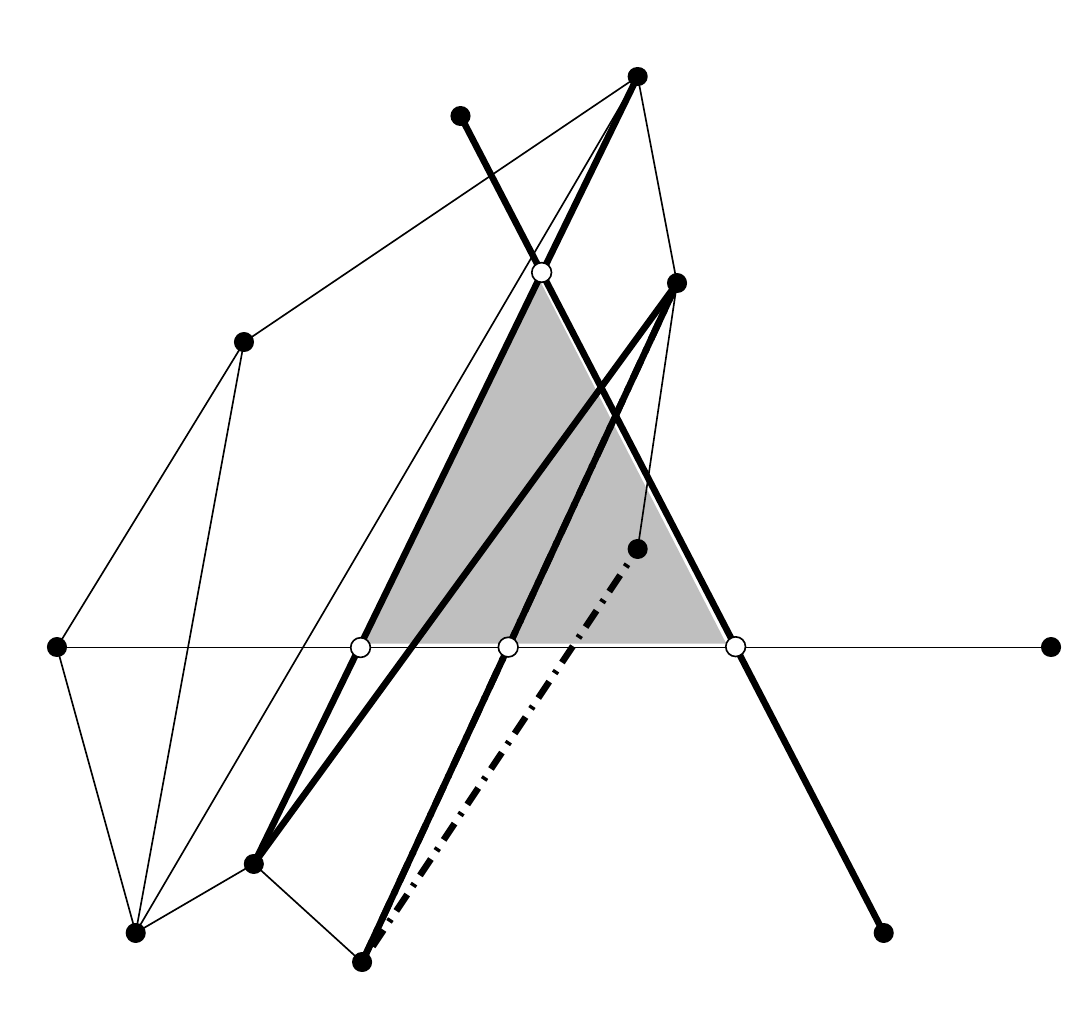}},{}]
We first show that triangle bounded by $d_i$, $d_j$ and $a b$ contains a vertex from the boundary of polygon $\polygon$.
Consider a maximal contiguous subsequence of edges of $d_1,\ldots,d_n$ that starts with $d_i$ and 
every edge in the subsequence intersects $d_j$.
Let $d_h=uv$ be the last edge in this subsequence.
Note that such $d_h$ is well-defined.
Let $u$ and $\vertex$ be the endpoints of $d_h$ below and above $a b$ respectively, 
and let $a_h$ be the intersection point of $d_h$ and $a b$.
Clearly $a_h$ is between $a_i$ and $a_j$. 
Based on the maximality of the subsequence $d_i,\ldots,d_h$, edge $d_h$ is the last edge that intersects $d_j$, so $d_{h+1}$ does not intersect $d_j$.
Also $d_{h+1}$ intersects $a b$ at a point between $a_h$ and $a_j$ and shares an endpoint with $d_h$. 
Edge $d_{h+1}$ cannot be incident on $\vertex$ because any edge connecting $\vertex$ to a point on line segment $a_h a_j$ 
intersects $d_j$. Therefore, $d_{h+1}$ is incident on $u$. Also, $d_{h+1}$ intersects $a_h a_j$, and it doesn't intersect 
$d_j$ or $d_{h}$. Thus, the other endpoint of $d_{h+1}$ must be in the triangle $\triangle a_i \ell a_j$.
Let $x$ be this endpoint.
\end{window}

\begin{window}[0,r,{\resizebox{2in}{1.8in}{\hspace*{0.5in}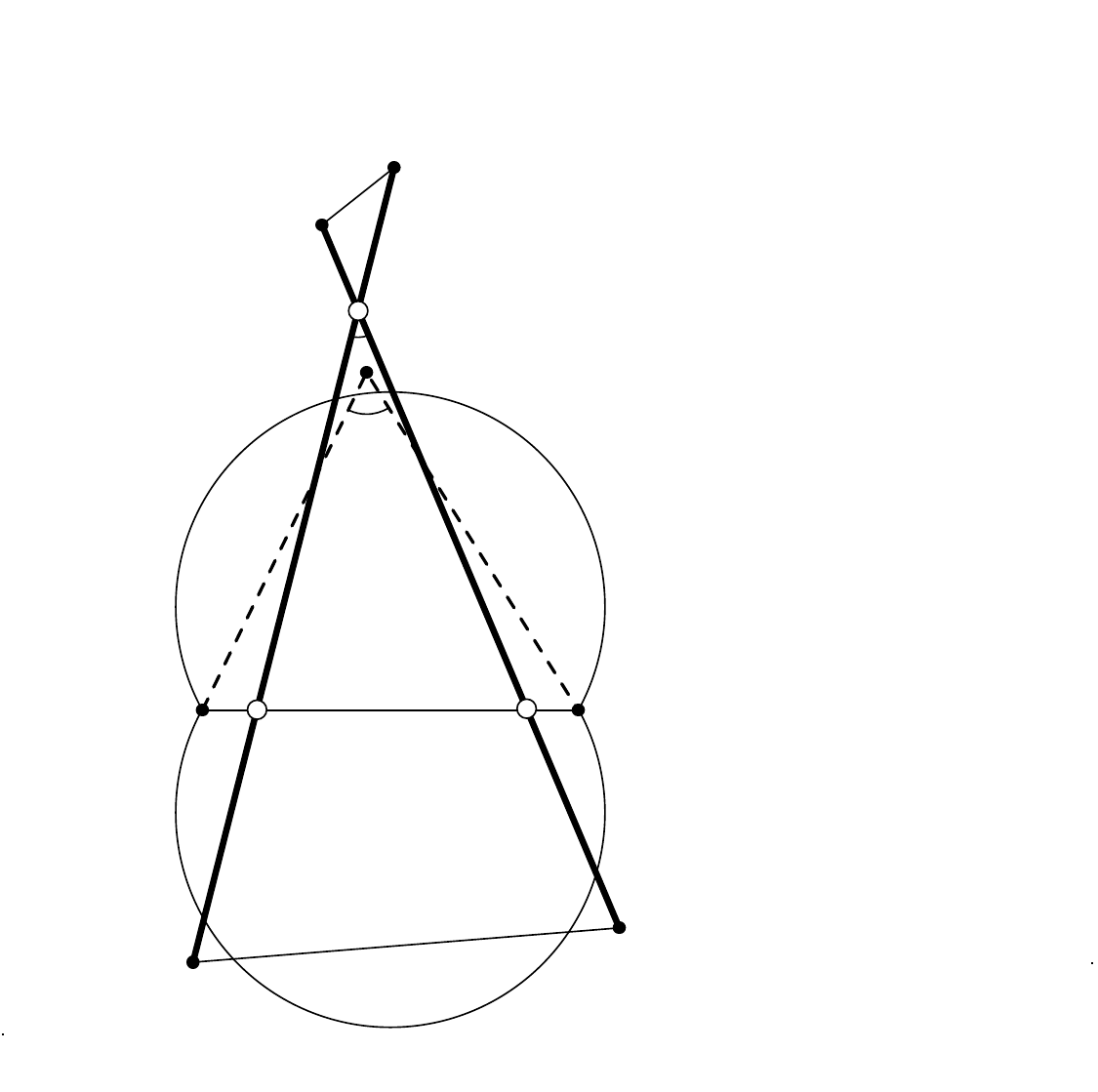}},{}]
The argument so far shows the existence of a point $x$ in the triangle bounded by $d_i$, $d_j$ and $a b$ (see the figure to the right). 
Thus, $\angle{a_i \ell a_j} \le \angle{axb} < \pi/3$.
The first inequality holds because $x$ is in triangle $\triangle a_i \ell a_j$, 
and the second inequality holds for the following reason. From the definition of $\beta$-skeleton, edge $a b$ is in the $\beta$-skeleton 
if and only if there does not exist a point $z$
in the point set such that $\angle a z b \ge \arcsin(1/\beta)$. For $\beta = 1/\sin (\pi/3.1)$, this implies that
$\angle a w b \le \pi/3.1$ for any point $w$ in the point set. 

\hspace*{\parindent}
In triangle $\triangle{q \ell r}$ the angle $\angle{q \ell r}$ is less than $\pi/3$, so 
$|q r| < \max\{|\ell q|,|\ell r|\} < \max\{|p q|,|r s|\}$, 
and a similar argument on triangle $\triangle{p \ell s}$ shows that $|p s|< \max\{|p q|,|r s|\}$.
\end{window}
This completes the proof of Fact~\ref{obs:point}.
\end{proof}

The remainder of the argument is similar to Case 1, where triangles do not cross.
The interior of polygon $\polygon$ can be retriangulated at lesser cost using the techniques
in the original $\beta$-skeleton arguments. 
Finally, lowering the weight of triangles in $t_1,\ldots,t_m$ by $\epsilon>0$
and raising the weight of new triangles by $\epsilon$,
gives a fractional triangulation that costs less than $\cost(\fracTriang^*)$.


\noindent
{\bf Part \ref{in:YXY} ($\YXY$-subgraph):}
 The argument used above for the $\beta$-skeleton works for the the $\YXY$-subgraph as well.
The only parts of the argument for the $\beta$-skeleton that use geometric properties of 
edges in the $\beta$-skeleton are Facts~\ref{prop:length} and~\ref{prop:remote}.
Thus, it suffices to show that these two facts hold for the edges of the $\YXY$-subgraph too.
\begin{window}[0,r,{\resizebox{2in}{1.8in}{\hspace*{0.5in}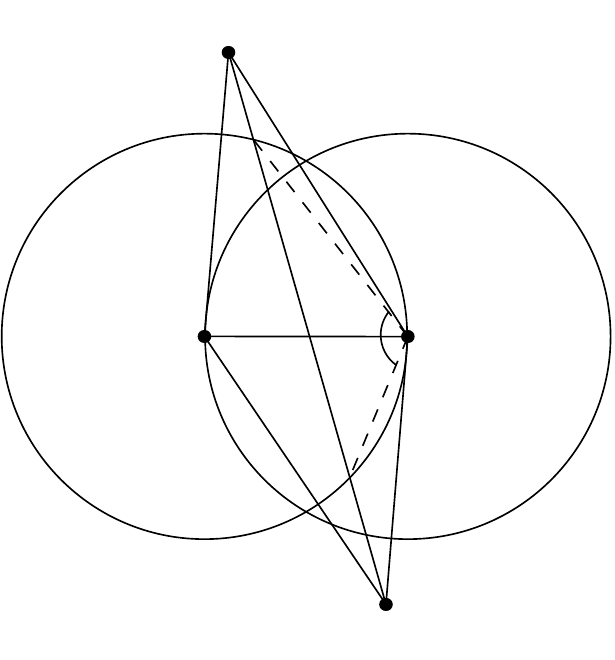}},{}]
Let $a b$ be an edge of the $\YXY$-subgraph and $p q$ be any edge that intersects $a b$. 
By definition of the $\YXY$-subgraph,
$|a b|$ is at most $\min\{|p a|,|p b|,|q a|,|q b|\}$, so
the union of two disks centered at $a$ and $b$ with radius $|a b|$ doesn't contain $p$ 
or $q$ (see the figure to the right). 
If the angle $\angle p a q$ in triangle $\triangle a p q$ and $\angle p b q$ in triangle $\triangle b p q$
are both greater than $90\dg$, we have $|p q| > \max\{|p a|,|p b|,|q a|,|q b|,|a b|\}$,
%
and proof of Fact~\ref{prop:length} is complete.
To show that $\angle p b q > 90\dg$, 
let $s$ and $t$ be the intersections of $pq$ with the circle centered at $a$.
We have $\angle p b q > \angle s b t > 90\dg$. The second inequality holds because
$\angle s b t$ is an inscribed angle (i.e. is an angle formed by two chords with a 
common endpoint in the circle centered at $a$) and its intercepted arc (i.e. the part of the circle 
which is "inside" the angle) is greater than $180\dg$. The argument to show $\angle p a q > 90\dg$
is similar.
\end{window}
We next show that Fact~\ref{prop:remote} 
also holds for any edge $a b$ of the $\YXY$-subgraph.
Let $p,q,r$ and $s$ be four other distinct points of the point set 
such that $p q$ and $r s$ both intersect $a b$.
We first consider the case that $p q$ and $r s$ do not intersect.
Let $h(q)$ and $h(r)$ respectively denote the distance of $q$ and $r$ from the line through $a b$. 
Lemma 2 in \cite{yang1994chain} states that if $h(q)\le h(r)$, then $|q r|<|r s|$. 
Similarly, if $h(r)\le h(q)$, then $|q r|<|p q|$. Hence $|q r| < \max\{|p q|, |r s|\}$
which proves Fact~\ref{prop:remote} in this case. 
If on the other hand, $p q$ and $r s$ intersect, we use Fact~\ref{obs:point}. The proof of 
Fact~\ref{obs:point} is almost the same whether $a b$ is an edge of $\YXY$-subgraph or an edge of
the $\beta$-skeleton.
%
%
\end{fullproof}


\begin{lemma}\label{lemma:out}
  If any of the following conditions holds for a potential edge $\edge$ of $\graph$ 
  (not on the boundary of the region to be triangulated),
  the LP forces $\edge$ to 0.
  \begin{enumerate}
  \item \label{out:independence} {\bf independence:}
    The LP forces a potential edge that crosses $\edge$ to 1.  
\item \label{out:diamond}
    {\bf diamond:} The diamond condition holds for $\edge$.
  \item \label{out:LMT} {\bf LMT skeleton:}
    For every two triangles $\tri$ and $\tri'$ for which $\edge$ is locally minimal,
    the LP forces one of the edges of $\tri$ or $\tri'$ to 0.
 \end{enumerate}
\end{lemma}

\begin{proofidea}
  Part~(\ref{out:independence}) is straightforward: if potential edges $\edge$ and $\edge'$ cross,
  then the LP covering constraint for a point near the intersection of $\edge'$ and $\edge$
  implies that the total weight of potential triangles that have $\edge$ or $\edge'$ as sides is at most 1.

  Part~(\ref{out:LMT}), the LMT skeleton, is straightforward.
  If an optimal fractional triangulation $\fracTriang^*$ gives $\edge$ positive weight,
  then (by constraint~(\ref{constraint:edge}) implied by the LP)
  there must be two triangles $\tri$ and $\tri'$ with positive $\fracTriang^*_\tri$ and $\fracTriang^*_{\tri'}$
  whose intersection is $\edge$.
  Edge $\edge$ must be locally minimal for $\tri$ and $\tri'$ 
  (otherwise $\fracTriang^*$ could be improved by reducing $\fracTriang^*_\tri$ and $\fracTriang^*_{\tri'}$
  and raising the weights of the other two triangles that triangulate $\tri\cup \tri'$).

  Part~(\ref{out:diamond}), the diamond condition, 
  is handled as the $\beta$-skeleton and $\YXY$-subgraph
  are handled in the proof idea of Lemma~\ref{lemma:in}.
\end{proofidea}


\begin{fullproof}
As discussed in the proof idea, Part \ref{out:independence} (independence) 
and Part \ref{out:LMT} (LMT skeleton) are straightforward. We give the detailed proof of 
Part \ref{out:diamond} (diamond).

\bigskip
\noindent
{\bf Part \ref{out:diamond} (diamond):}
Like $\beta$-skeleton and YXY subgraph, 
the original diamond heuristic for MWT is proved by local-improvement arguments:
if the condition holds for an edge $\edge$ that {\em is} in the MWT,
then a polygon covering $\edge$ within the MWT can be retriangulated at lesser cost,
contradicting the optimality of the MWT \cite{das1989triangulations,drysdale2001exclusion}.
We first give a summary of the results in \cite{drysdale2001exclusion} 
and then use them to extend the result to any optimal {\em fractional} triangulation $\fracTriang^*$.

\begin{window}[2,r,{\resizebox{2in}{1.8in}{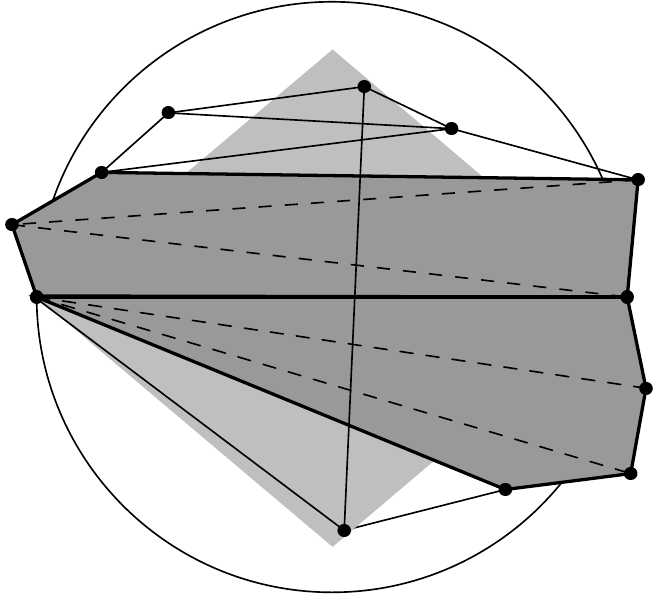}},{}]
Suppose $\edge$ is horizontal and $p$ and $q$ are its endpoints, and $p$ is on the left of $q$. 
Let $\triangle p q w$ and $\triangle p q u$ be the two isosceles triangles with base angle $\pi/4.6$ 
above and below $p q$ and $C$ be the disk with diameter $\edge$ as shown in the figure to the right.
Suppose that $\triangle p q w$ contains a point $a'$ and $\triangle p q u$ contains a point $b'$.
If $\edge$ is in the MWT, $a' b'$ is not in the MWT, and there is a set of triangles in the MWT that 
intersect $a' b'$. Consider the sequence of triangles 
encountered when tracing $a' b'$ toward $a'$, starting from edge $\edge$ and stopping with the 
first triangle that has a vertex inside disk $C$, and let $P_1$ be the polygon formed by the boundary
edges of triangles in the sequence.
Let $a$ be the vertex found inside $C$ --- if all else fails, then $a= a'$.
In the preceding figure, $P_1$ is the dark gray area above $\edge$.  
Similarly, consider the sequence of triangles 
encountered when tracing $a' b'$ toward $b'$, starting from edge $\edge$ and stopping with the 
first triangle that has a vertex inside disk $C$, and let $P_2$ be the polygon formed by the boundary
edges of triangles in the sequence. In the figure, $P_2$ is the dark gray area below $\edge$.
The boundary edges of $P_1$ are grouped naturally into two chains, one from $p$ to $a$ and 
one from $q$ to $a$. Vertex $a$ doesn't belong to any of the two chains. 
$P_1$ is a {\em fan} on $p$ (or $q$) if all triangles in $P_1$ are incident on $p$ (or $q$). 
Similarly $P_2$ can be a fan on $p$ (or $q$).
\end{window}

Drysdale et al.~\cite{drysdale2001exclusion} prove the following two facts to show how $P_1$ or $P_2$ or 
their union can be triangulated at lesser cost.\footnote
{In summarizing the result of \cite{drysdale2001exclusion}, we use the same notations and names. 
The only exception is that \cite{drysdale2001exclusion} uses $A$ and $B$ instead of $P_1$ and $P_2$.}   
\begin{fact}[following {\cite[Lemma 8]{drysdale2001exclusion}}]\label{prop:diamond8}
If $P_1$ (or $P_2$) is not a fan, it can be retriangulated at lower cost. 
\end{fact}

\begin{fact}[{\cite[Lemma 9]{drysdale2001exclusion}}]\label{prop:diamond9}
When both $P_1$ and $P_2$ are fans, then their union can be retriangulated at lower cost.
\end{fact}

The above facts contradict the optimality of the MWT. Thus, $\edge$ cannot be in any MWT.
Next we extend the above arguments to any optimal {\em fractional} triangulation $\fracTriang^*$.
We find polygons $P'_1$ and $P'_2$ corresponding to $P_1$ and $P_2$ in the above argument and show that if $\fracTriang^*$ does not give 
$\edge$ weight zero, then $P'_1\cup P'_2$ can be retriangulated at lesser cost. Lowering the weight of those
triangles by $\epsilon> 0$ and raising the weight of other triangles by $\epsilon$ gives a fractional triangulation
that costs less than $\fracTriang^*$. The details of the argument follow.
 
If $\fracTriang^*$ does not give $p q$ weight 0, there is a triangle $t$ with positive $\fracTriang^*_t$ 
that has $p q$ as a side. Triangle $t$ intersects $a' b'$.
For any side $d$ of $\tri$ that intersects $a' b'$ there must be a triangle $\tri'$ with positive $\fracTriang^*_{\tri'}$ 
that has $d$ as a side and lies on the other side of $d$ from $\tri$. By repeating the same argument 
for the new triangle(s), a sequence $\Gamma$ of triangles can be obtained that completely covers $a' b'$. 
As shown in the figure below triangles in the sequence $\Gamma$ may or may not cross each other.
The left figure shows the case that no triangles cross, while the right figure shows the case that some triangles cross. 
If triangles in the sequence do not cross, the original arguments from \cite{drysdale2001exclusion} apply.
However, if triangles in the sequence cross, additional analysis is needed.

\begin{center}
\scalebox{.65}
{
\xfig{sequence}
}
\end{center}
Consider the set of triangles in $\Gamma$ 
encountered when tracing $a' b'$ toward $a'$, starting from edge $\edge$ and stopping with the 
first triangle that has a vertex inside disk $C$, and let $P'_1$ be the polygon formed by the 
boundary edges of these triangles. Also let
$a$ be the vertex found inside $C$ --- if all else fails, then $a= a'$.
Similarly, define $P'_2$ to be the polygon formed by the boundary edges of the set of triangles encountered 
when tracing $a' b'$ from $\edge$ toward $b'$ until a 
vertex is inside $C$. 
The following figure shows $P'_1$ and $P'_2$ (the dark shaded regions) in two cases.
In the left figure there are no crossing triangles while in the right figure some triangles in $P'_1$ cross. 
\begin{center}
\scalebox{.7}
{
  \xfig{ABregion}
}
\end{center}

We consider the following cases and show how in each case $P'_1\cup P'_2$ can be retriangulated at lower cost.


\begin{window}[0,r,{\resizebox{2.7in}{1.55in}{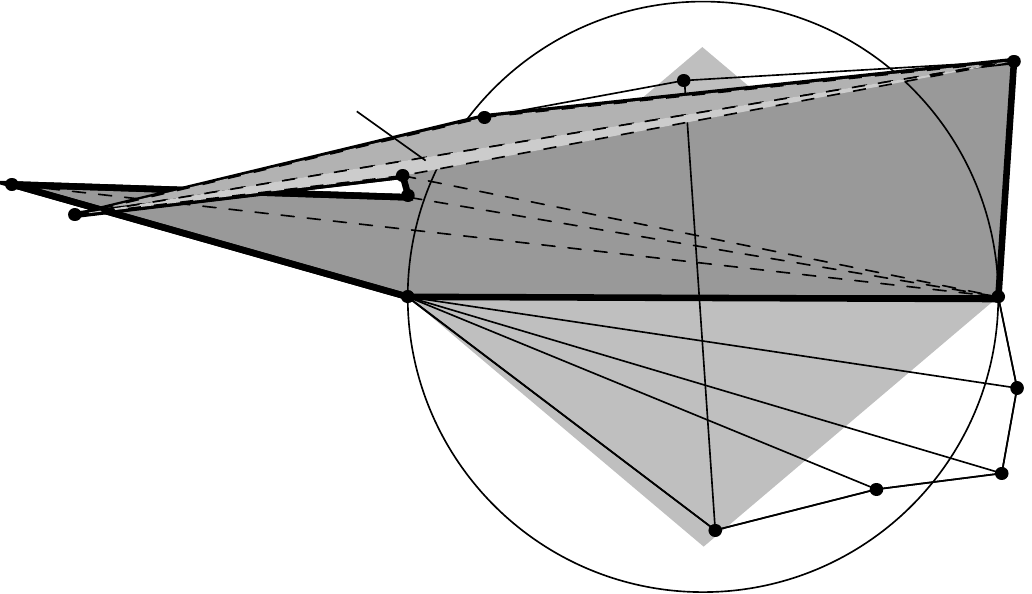}},{}] 
\noindent{\em Case 1} --- $P'_1$ is not a fan and some triangles in $P'_1$ overlap.
This case is shown in the figure to the right and its magnified version below.
In the sequence of triangles in $P'_1$ 
encountered when tracing $a' b'$ toward $a'$, let $\tri'$ be the first triangle 
that crosses some of the previous triangles in the sequence (the light gray triangle labeled in the figure).
Let $P''$ be the polygon that covers the set of all triangles before $\tri'$ (the dark gray polygon labeled $P''$ in the figure).
Polygon $P''$ and triangle $\tri'$ share an edge. Let $p' q'$ be this edge such that $p'$ is to left of $q'$ (see the figure below).
The boundary edges of $P''$ can be grouped naturally into two chains, one from $p$ to $p'$ and 
one from $q$ to $q'$.
We use the following additional facts from \cite{drysdale2001exclusion} to show that $P''$ can be retriangulated at lower cost.
Note that the following facts can be applied to $P''$ because 
$P''$ is the union of triangles that do not cross.
\end{window}

\begin{fact}[{\cite[Lemma 7]{drysdale2001exclusion}}]\label{lemma:diamond7}
If the chain of boundary vertices from $p$ to $p'$ 
has three consecutive vertices $x$, $y$ and $z$ with $|y q|\ge|z q|$ 
and the internal angle $\angle x y z$ is less than $\pi$, then $P''$ can be retriangulated 
to decrease its cost. 
(The same is true with $q$ and $q'$ exchanging roles with $p$ and $p'$, respectively.)
\end{fact}

We say the {\em clockwise limit \footnote{ We use the term {\em clockwise limit} to be consistent with the 
terminology of the original paper \cite{drysdale2001exclusion}. The word {\em limit} in this context has nothing 
to do with the well-known mathematical limit.} on the directions of the boundary edges on the chain} (from $p$ to $p'$)
is perpendicular to $q w$, if for every two consecutive vertices $p_i$ and $p_{i+1}$ on the chain,
$p_{i+1}$ is above or on the line through $p_i$ that is perpendicular to $q w$.

\begin{fact}[{\cite[Lemma 6]{drysdale2001exclusion}}]\label{lemma:diamond6}
If the chain of boundary vertices from $p$ to $p'$ 
has no 
three consecutive vertices $x$, $y$, and $z$ with $|y q|\ge|z q|$ 
that form an internal angle ($\angle{x y z}$) of less than $\pi$, 
then the clockwise 
limit on the directions
of the boundary edges is perpendicular to $q w$. 
(The same is true with $q$ and $q'$ exchanging roles with $p$ and $p'$, respectively,
and ``counter-clockwise'' replacing ``clockwise''.)
\end{fact}


Facts~\ref{lemma:diamond6} implies that if the clockwise limit on the direction 
of boundary edges (on the chain from $p$ to $p'$) is not perpendicular to the line through $q w$, then
there are three consecutive vertices $x$, $y$, and $z$ with $|y q|>|z q|$ 
that form an internal angle ($\angle{x y z}$) of less than $\pi$, and the existence of 
such vertices by Fact~\ref{lemma:diamond7} implies that $P''$ can be retriangulated 
to decrease its cost. The remainder of the proof for Case 1 shows that one of the edges 
on the boundary of $P''$ violates the mentioned
limit on the direction of boundary edges, and thus, $P''$ can be retriangulated at lower cost.


In the following figure, the darker shaded area is $P''$. Triangles in $P''$ do not overlap. 
Triangle $\tri' = \triangle p' q' \vertex$ is the first triangle that overlaps some of the previous triangles.\label{case2}
Triangle $\tri'$ and polygon $P''$ share edge $p' q'$,
and $\tri'$ crosses some triangles covered by $P''$. Thus, $\tri'$ crosses some boundary 
edges of $P''$ either on the chain from $p$ to $p'$ or on the chain from $q$ to $q'$.
Assume without loss of generality that triangle $\tri'$ crosses an edge on the chain from $p$ to $p'$.
Since edge $p' q'$ is a boundary edge of $P''$, the other two sides of triangle $\tri'$ ($ p' \vertex$ and $q' \vertex$) should intersect $P''$.
Let $p_ip_{i+1}$ be the first edge on the chain from $p$ to $p'$ that is intersected by $q' \vertex$ when moving from $q'$ to $\vertex$.

\begin{center}
\scalebox{.5}
{
  \xfig{case2large}
}
\end{center}
The line through $q' \vertex$ divides the plane into a half-space above it (the half-space containing point $w$) and 
a half-space below it. It's easy to see that $p_i$ is in the half-space above $q' \vertex$ and $p_{i+1}$ is below it, 
so $p_{i+1}$ cannot be above the line through $p_i$ that is perpendicular to $q w$. This means that the clockwise 
limit on the direction of the boundary edges is not perpendicular to $q w$. 
This combined with Fact~\ref{lemma:diamond7} and Fact~\ref{lemma:diamond6} implies that the interior of $P''$ can 
be retriangulated at lower cost.


\noindent{\em Case 2 ---} $P'_1$ is not a fan and no two triangles in $P'_1$ cross. 
Since triangles in $P'_1$ do not overlap, Fact~\ref{prop:diamond8} directly implies that $P'_1$ can be retriangulated at a lesser cost.

The previous two cases show that if $P'_1$ is not a fan, we can retriangulate some triangles in $P'_1$ to reduce the cost of triangulation.
A similar argument applies to $P'_2$, so there remains the case where both $P'_1$ and $P'_2$ are fans. We consider this case next.


\noindent{\em Case 3 ---} $P'_1$ and $P'_2$ are both fans. 
$P'_1$ is a fan on $p$ (or $q$) if all triangle in $P'_1$  are incident on $p$ (or $q$), and 
thus no two triangle in $P'_1$ overlap. Similarly if $P'_2$ is a fan, no two triangle in $P'_2$
overlap. Additionally, when both $P'_1$ and $P'_2$ are fans, triangles in $P'_1$ are all above the
line through $pq$ and triangles in $P'_2$ are below the line through $pq$, and no triangle in $P'_1$
intersects a triangle in $P'_2$. Hence, no two triangle in $P'_1\cup P'_2$ can cross, and 
by Fact~\ref{prop:diamond9} $P'_1\cup P'_2$ can be retriangulated at lower cost.
%

In all the above cases,
by lowering the weight of some triangles in $P'_1\cup P'_2$ by $\epsilon>0$
and raising the weight of some other triangles in the new triangulation by $\epsilon$, 
a fractional triangulation costing less than $\fracTriang^*$ can be obtained, which contradicts the optimality of $\fracTriang^*$.
This completes the proof of Part \ref{out:diamond} (diamond property) and Lemma~\ref{lemma:out}. 
\end{fullproof}



Assume (as in the statement of Thm.~\ref{thm:heuristics}) that the set $\edges^*$ of edges
that can be deduced to be in every MWT of $\graph$ 
gives a partition of $\graph$ in which every face is empty.
It follows from Lemmas \ref{lemma:in} and \ref{lemma:out}
(by a simple inductive proof) that every edge that can be deduced to be excluded from every MWT
is forced to 0 by the LP, and every edge that can be deduced to be in every MWT is forced to 1.
Thus, in any optimal fractional triangulation $\fracTriang^*$,
no potential triangle $\tri$ that crosses an edge in $\edges^*$ has positive weight $\fracTriang^*_\tri$.
Thus, the optimal fractional triangulations $\fracTriang^*$ are exactly those that,
for each face $\face$ of the partition, induce an optimal fractional triangulation 
of the simple polygon $\face$.
It is known
(e.g.~\cite[Thm.~7]{dantzig1985triangulations},
\cite[Thm.~4.1(i)]{de1996polytope},
\cite[Cor.~3.6.2]{kirsanov2004minimal})
that, for any simple polygon $\face$, each basic optimal fractional triangulation
is the incidence vector of an actual triangulation of $\face$.
Thus, each optimal extreme point of the LP for $\graph$ is also the incidence vector
of a triangulation of $\graph$, proving Thm.~\ref{thm:heuristics}.

\section{Acknowledgements}

Thanks to two anonymous referees for suggestions on improving the presentation of the results.

\bibliographystyle{siam}
\bibliography{bib}

\end{document}